\documentclass[AMA,STIX1COL]{WileyNJD-v2}
\usepackage{moreverb}

\def\cl{{C}\!\ell}
\def\R{{\mathbb R}}
\def\F{{\mathbb F}}
\def\CC{{\mathbb C}}
\def\C{\mathcal {G}}
\def\P{{\rm P}}
\def\A{{\rm A}}
\def\B{{\rm B}}
\def\Q{{\rm Q}}
\def\Z{{\rm Z}}
\def\H{{\rm H}}
\def\G{{\rm G}}
\def\Aut{{\rm Aut}}
\def\ker{{\rm ker}}
\def\OO{{\rm O}}
\def\SO{{\rm SO}}
\def\Pin{{\rm Pin}}
\def\Spin{{\rm Spin}}
\def\ad{{\rm ad}}
\def\mod{{\rm \;mod\; }}

\newcommand\BibTeX{{\rmfamily B\kern-.05em \textsc{i\kern-.025em b}\kern-.08em
T\kern-.1667em\lower.7ex\hbox{E}\kern-.125emX}}

\articletype{Special issue paper}

\received{<day> <Month>, <year>}
\revised{<day> <Month>, <year>}
\accepted{<day> <Month>, <year>}

\begin{document}

\title{On generalization of Lipschitz groups and spin groups\protect\thanks{The publication was prepared within the framework of the Academic Fund Program at the HSE University in 2022 (grant 22-00-001).}}

\author[1]{Ekaterina Filimoshina}

\author[2,3]{Dmitry Shirokov*}

\authormark{EKATERINA FILIMOSHINA AND DMITRY SHIROKOV}

\address[1]{ \orgname{HSE University}, \orgaddress{\state{Moscow}, \country{Russia}}}

\address[2]{ \orgname{HSE University}, \orgaddress{\state{Moscow}, \country{Russia}}}

\address[3]{ \orgname{Institute for Information Transmission Problems of Russian Academy of Sciences}, \orgaddress{\state{Moscow}, \country{Russia}}}

\corres{Dmitry Shirokov. \email{dm.shirokov@gmail.com}}

\presentaddress{HSE University, 101000, Moscow, Russia}

\abstract[Abstract]{This paper presents some new Lie groups preserving fixed subspaces of geometric algebras (or Clifford algebras) under the twisted adjoint representation.
We consider the cases of subspaces of fixed grades and subspaces determined by the grade involution and the reversion.
Some of the considered Lie groups can be interpreted as generalizations of Lipschitz groups and spin groups. 
The Lipschitz groups and the spin groups are subgroups of these Lie groups and coincide with them in the cases of small dimensions. 
We study the corresponding Lie algebras.}



\maketitle

\section{Introduction}

Let us consider the real geometric algebra \cite{hestenes, LM, lg1, lounesto} (or the Clifford algebra) $\C_{p,q}=\cl_{p,q}=\cl(\R^{p,q})$, $p+q=n\geq1$, or the complex  geometric algebra $\cl(\CC^n)$, $n\geq1$. When considering both of these cases,
we denote the corresponding algebra by $\C$.

We denote the identity element of the algebra $\C$ by $e$, the generators by $e_a$, $a=1, \ldots, n$. In the case of the real geometric algebra $\C_{p,q}$, the generators satisfy 
\begin{eqnarray}
e_a e_b+e_b e_a=2\eta_{ab} e,\qquad a, b=1, \ldots, n,\label{gener}
\end{eqnarray}
where $\eta=(\eta_{ab})$ is the diagonal matrix with $p$ times $1$ and $q$ times $-1$ on the diagonal. In the case of the complex geometric algebra $\cl(\CC^n)$, the generators satisfy the same conditions (\ref{gener}) but with the identity matrix $\eta=I_n$ of size $n$. The other basis elements of $\C$ are products of the generators $e_{a_1 \ldots a_k}:=e_{a_1}\cdots e_{a_k}$, $a_1 \leq \cdots \leq a_k$. An arbitrary element (multivector)  $U\in\C$ has the form
\begin{eqnarray*}
U=ue+\sum_{a=1}^n u_a e_a+\sum_{a<b}u_{ab}e_{ab}+\cdots+ u_{1\ldots n}e_{1\ldots n},\qquad u, u_a, \ldots, u_{1\ldots n}\in\F.
\end{eqnarray*}
We use $\F$ to denote the field of real numbers $\R$ in the case $\C_{p,q}$ and the field of complex numbers $\CC$ in the case $\cl(\CC^n)$.

The grade involute of the element $U\in\C$ is denoted by $\widehat{U}$, the reverse of $U\in\C$ is denoted by $\widetilde{U}$. The superposition of the grade involution and the reversion
is usually called Clifford conjugation. In this paper, we do not use a distinct notation for the Clifford conjugation, denoting it by the combination of the symbols of the grade involution and the reversion $\widehat{\widetilde{U}}$. We have the following well-known properties of these three operations
\begin{eqnarray}\label{^uv=^u^v}
\widehat{UV}=\widehat{U}\widehat{V},\qquad \widetilde{UV}=\widetilde{V}\widetilde{U},\qquad \widehat{\widetilde{UV}}=\widehat{\widetilde{V}}\widehat{\widetilde{U}},\qquad \forall U, V\in \C.
\end{eqnarray}

Consider the subspaces of fixed grades $\C^k$, $k=0, \ldots, n$, which elements are linear combinations of the basis elements $e_{a_1 \ldots a_k}$ with multi-indices of length $k$. Also consider the even $\C^{(0)}$ and the odd $\C^{(1)}$ subspaces
$$
\C^{(k)}=\{U\in\C:\quad \widehat{U}=(-1)^k U\}=\bigoplus_{j=k \mod 2} \C^j,\qquad k=0, 1.
$$
We call the elements of these subspaces even elements and odd elements respectively.
The product of the elements of the same parity is even, the product of the elements of different parities is odd, i.e.
\begin{eqnarray}
\C^{(k)} \C^{(l)}\subset \C^{(k+l)\mod 2},\qquad k, l=0, 1.\label{even}
\end{eqnarray}
Also we consider the four subspaces determined by the grade involution and the reversion (they are called the subspaces of quaternion types $0, 1, 2$, and $3$ 
\cite{quat1, quat2, quat3})
\begin{eqnarray}\label{qtdef}
\C^{\overline{k}}=\{U\in\C:\quad \widehat{U}=(-1)^k U,\quad \widetilde{U}=(-1)^{\frac{k(k-1)}{2}} U\}=\bigoplus_{j=k \mod 4} \C^j,\qquad k=0, 1, 2, 3.
\end{eqnarray}
We use the upper multi-index instead of the direct sum symbol in order to denote the direct sum of different subspaces. For example, $\C^{(1)\overline{23}4}:=\C^{(1)}\oplus\C^{\overline 2}\oplus \C^{\overline{3}}\oplus\C^{4}$.
We denote the center of the algebra $\C$ by
\begin{eqnarray}
\Z&:=&\left\lbrace
\begin{array}{lll}
\C^0&\mbox{if $n$ is even},&
\\
\C^0\oplus\C^n&\mbox{if $n$ is odd}.&
\end{array}
\right.\label{center}
\end{eqnarray}
Consider the group of invertible elements of the algebra $\C$
\begin{eqnarray*}
\C^{\times}&:=&\{T\in\C:\quad\exists T^{-1}\}.
\end{eqnarray*}
In this paper, the subset of invertible elements of any set is denoted with $\times$. 

Consider the twisted adjoint representation $\check{\ad}$ acting on the group $\C^\times$
\begin{eqnarray*}
\check{\ad}:\C^\times \to \Aut\C
\end{eqnarray*}
as $T \mapsto \check{\ad}_T$, where $\check{\ad}_{T}(U)=\widehat{T}U T^{-1}$ for any $U\in\C$. 
It differs from the ordinary adjoint representation by the operation of grade involution.
We use the following notation for the regular adjoint representation 
\begin{eqnarray*}
\ad:\C^\times \to \Aut\C,
\end{eqnarray*}
where $T \mapsto \ad_T$ and $\ad_{T}(U)=TU T^{-1}$ for any $U\in\C$.

The twisted adjoint representation was introduced in a classic paper of three authors \cite{ABS}, and it is an important mathematical notion with respect to both algebra and geometry. When applying this representation to the vector $v\in\C^1$, it is the reflection across the hyperplane orthogonal to the vector $s\in\C^{\times 1}$: 
$$\check{\ad}_s (v)=\widehat{s} v s^{-1}=v-2\frac{q(v,s)}{q(s,s)}s,$$
where $q(x,y)=\frac{1}{2}(xy+yx)$, $x, y\in\C^1$, is a symmetric bilinear form over $\C^1$.
In our notation, we have $q(e_a, e_b)=\eta_{ab}$ for the matrix $\eta$ (\ref{gener}).
In the case of arbitrary dimension $n$ and signature $(p,q)$, the representation  $\check{\ad}$ is a double covering of the orthogonal groups $\OO(p,q)$ by the corresponding spin groups $\Pin(p,q)$.
Namely, for any matrix $P=(p_a^b)\in\OO(p,q)$, there exist exactly two elements $\pm T\in\Pin(p,q)$ in the corresponding spin group that are related in the following way:
$$\widehat{T} e_a T^{-1}=p_a^b e_b.$$

The following sets are the kernels of the considered representations
\begin{eqnarray}
\ker(\ad)&=&\{T\in\C^{\times}:\quad \ad_{T}(U)=U\quad\forall U\}=\Z^\times,\label{kerad}\\
\ker(\check{\ad})&=&\{T\in\C^{\times}:\quad \check{\ad}_{T}(U)=U\quad\forall U\}=\C^{0 \times},\label{kerchad}
\end{eqnarray}
where $\Z^\times$ is the set of invertible elements of the center (\ref{center}) and $\C^{0 \times}$ is the set of invertible elements of grade $0$.

In the previous paper of one of the authors \cite{OnInner}, there were introduced and studied the Lie groups preserving the fixed subspaces 
\begin{eqnarray}
\C^k,\quad k=0, 1, \ldots, n;\qquad \C^{(j)},\quad j=0, 1;\qquad \C^{\overline{m}},\quad \C^{\overline{ms}},\quad m, s=0, 1, 2, 3 \label{subspaces}
\end{eqnarray}
under the adjoint representation $\ad$.

In the present paper, we consider the similar issue regarding the twisted adjoint representation $\check{\ad}$, which is interesting for studying spin groups, Lipschitz groups, and their generalizations. Namely, in this paper, we introduce and study the Lie groups preserving the subspaces (\ref{subspaces}) under the twisted adjoint representation $\check{\ad}$. We also  introduce the generalized spin groups based on the mentioned groups. We study the corresponding Lie algebras.

Note that the results  of the paper \cite{OnInner} on inner automorphisms can be generalized to the case of graded central simple algebras \cite{Wall} and graded central simple algebras with involution \cite{Wall2}. In the current paper, we do not discuss the relations with these more general algebras because of the specificity of the twisted adjoint representation $\check{\ad}$ for geometric (Clifford) algebras. 

Let us introduce the following notation.
We denote the groups of elements preserving the subspaces (\ref{subspaces}) under the representation $\ad$ by (as in the paper \cite{OnInner})
\begin{eqnarray}
\Gamma^k,\quad k=0, 1, \ldots, n;\qquad \Gamma^{(j)},\quad j=0, 1;\qquad \Gamma^{\overline{m}},\quad \Gamma^{\overline{ms}},\quad m, s=0, 1, 2, 3\label{Gamma}
\end{eqnarray}
respectively.
For example, we have $\Gamma^{(0)}=\{T \in \C^\times:\quad T \C^{(0)}T^{-1}\subseteq \C^{(0)}\}$. We denote with $\check{\quad}$ the groups of elements preserving the subspaces (\ref{subspaces}) under the twisted adjoint representation $\check{\ad}$ 
\begin{eqnarray}
\check{\Gamma}^k,\quad k=0, 1, \ldots, n;\qquad \check{\Gamma}^{(j)},\quad j=0, 1;\qquad \check{\Gamma}^{\overline{m}},\quad \check{\Gamma}^{\overline{ms}},\quad m, s=0, 1, 2, 3\label{gammacheck}
\end{eqnarray}
respectively.
For example, we have $\check{\Gamma}^{\overline{1}}=\{T \in \C^\times:\quad \widehat{T} \C^{\overline{1}}T^{-1}\subseteq \C^{\overline{1}}\}$.

In the paper \cite{OnInner}, there were also introduced the groups $\P$, $\A$, $\B$, and $\Q$.
In the present paper, we consider the analogues of these groups and denote them by $\P^{\pm}$, $\A_{\pm}$, $\B_{\pm}$, $\Q^{\pm}$  respectively
(the necessity of considering these groups is discussed in the next sections of this paper):
\begin{eqnarray}
\P=\Z^{\times}(\C^{(0)\times}\cup\C^{(1)\times}),&\qquad&
\P^{\pm}=\C^{(0)\times}\cup\C^{(1)\times}\subseteq \P,\label{P}\\
\A=\{T\in\C^\times:\quad\widetilde{T}T\in\Z^{\times}\},&\qquad&
\A_{\pm}=\{T\in\C^\times:\quad\widetilde{T}T\in \C^{ 0\times}\}\subseteq \A,\label{A}
\\
\B=\{T\in\C^\times:\quad\widehat{\widetilde{T}}T\in\Z^{\times}\},&\qquad&
\B_{\pm}=\{T\in\C^\times:\quad\widehat{\widetilde{T}}T\in \C^{0\times }\}\subseteq \B,\label{B}\\
\Q=\{T\in\Z^{\times}(\C^{(0)\times}\cup\C^{(1)\times}):\quad \widetilde{T}T\in\Z^{\times}\},&\qquad& \Q^{\pm}=\{T\in\C^{(0)\times}\cup\C^{(1)\times}:\quad \widetilde{T}T\in\C^{ 0\times}\}\subseteq \Q.\label{Q}
\end{eqnarray}
The upper index $\pm$ in (\ref{P}) and (\ref{Q}) means that the corresponding groups consist only of even and odd elements. The lower index $\pm$ in (\ref{A}) and (\ref{B}) means that the corresponding groups consist only of elements with positive and negative values of the corresponding norm functions $\widetilde{T}T$ and $\widehat{\widetilde{T}}T$. These notations are convenient for considering generalized spin groups (see Section~\ref{sectGS}).

Let us note that the two groups of the types (\ref{Gamma}) and (\ref{gammacheck}) respectively, which are considered in this paper, are well known.
One of them is the group $\Gamma^{1}$, which is called the Clifford group and often denoted by $\Gamma$ in the literature.
The elements of this group preserve the subspace of grade $1$ under the adjoint representation. 
There are the following equivalent definitions of this group
\begin{eqnarray}
\Gamma:=\Gamma^1&=&\{T\in\C^{\times}:\quad T\C^{1}T^{-1}\subseteq\C^{1}\}
\\
&=&\{T\in\Z^{\times}(\C^{(0)\times}\cup\C^{(1)\times}):\quad T\C^{1}T^{-1}\subseteq\C^{1}\}\label{Cl2}
\\
&=& \{W v_1\cdots v_m:\quad m\leq n,\quad W\in\Z^{\times},\quad v_j\in\C^{1\times}\}.
\end{eqnarray}
The group $\check{\Gamma}^1$ is also well known and is often called the Lipschitz group. 
We have the following (well-known) equivalent definitions of this group 
\begin{eqnarray}
\Gamma^\pm:=\check{\Gamma}^{1}&=& \{ T\in\C^\times:\quad \widehat{T} \C^1 T^{-1}\subseteq \C^1\}\label{Lip1}\\
&=&\{T\in\C^{(0)\times}\cup\C^{(1)\times}:\quad T \C^1 T^{-1}\subseteq \C^1\}\label{Lip2}\\
&=&\{ v_1 \cdots v_m:\quad m\leq n,\quad v_j\in\C^{ 1\times}\}.\label{Lip3}
\end{eqnarray}
Let us note that the equivalence of the first two definitions (\ref{Lip1}) and (\ref{Lip2}) of the Lipschitz group follows from the inclusion $\check{\Gamma}^1 \subseteq \P^{\pm}$ and is the particular case of a more general statement about the equivalence of the two definitions (\ref{x10}) and (\ref{eqgk'}) of the groups $\check{\Gamma}^k$, $k=0, 1, \ldots, n$ (see Lemma \ref{lemmaGammaP} and the notes after it). The equivalence of these two definitions to the third one (\ref{Lip3}) is usually proved using the Cartan–Dieudonné theorem \cite{Dieud}, we omit this fact in the present paper. Note that the Lipschitz group is usually denoted by $\Gamma^{\pm}_{p,q}$ in the real case, see \cite{lg1} (or just $\Gamma_{p,q}$, see \cite{lounesto}). The symbol $\pm$ means that the elements of the Lipschitz group are either even or odd and cannot be the sum of both even and odd elements; this corresponds to the definition (\ref{Lip2}).

The rest of the groups (\ref{gammacheck}) and the groups $\P^\pm$, $\A_\pm$, $\B_\pm$, $\Q^\pm$ (\ref{P}) - (\ref{Q}) can be interpreted as the \textit{generalizations of the Lipschitz group} $\check{\Gamma}^1$, as they preserve the other fixed subspaces of the algebra $\C$, different from $\C^1$, under the twisted adjoint representation $\check{\ad}$. We study these groups in detail in Sections \ref{sectionP*} - \ref{sectSmall}. We define \textit{generalized spin groups} as normalized subgroups of generalized Lipschitz groups and discuss them in Section \ref{sectGS}. We study the corresponding Lie algebras in Section \ref{liealg}. The results of Sections \ref{sectionP*} - \ref{liealg} are new. The conclusions follow in Section \ref{sectConcl}.

\section{The groups $\P^{\pm}$, $\check{\Gamma}^0$, $\check{\Gamma}^{\lowercase{n}}$, $\check{\Gamma}^{0\lowercase{n}}$, $\check{\Gamma}^{(0)}$, and $\check{\Gamma}^{(1)}$}\label{sectionP*}

Let us consider the group consisting of even and odd invertible elements
\begin{eqnarray}
\P^{\pm}&:=&\C^{(0)\times}\cup\C^{(1)\times}.
\end{eqnarray}
The groups $\P$ and $\P^{\pm}$ (\ref{P}) are related in the following way
\begin{eqnarray}\label{rel_P}
\P=\Z^{\times}\P^{\pm},
\end{eqnarray}
in particular, the groups $\P$ and $\P^{\pm}$ coincide if $n$ is even
\begin{eqnarray}\label{pp'g0g1g0'g1'}
\P^{\pm}=\P,\qquad\mbox{$n$ is even}.
\end{eqnarray}

In Section \ref{sectiongk'}, we consider the groups $\check{\Gamma}^k$, $k=0,1,\ldots,n$, preserving the subspaces of fixed grade $k$ under the twisted adjoint representation. Let us consider the two of these groups in this section:
\begin{eqnarray*}
\check{\Gamma}^{0}&:=&\{T\in\C^{\times}:\quad \widehat{T}\C^{0}T^{-1}\subseteq\C^{0} \},
\\
\check{\Gamma}^{n}&:=&\{T\in\C^{\times}:\quad \widehat{T}\C^{n}T^{-1}\subseteq\C^{n} \}.
\end{eqnarray*}

\begin{lemma}\label{lemmag0'gn'}
We have
\begin{eqnarray*}
\check{\Gamma}^{0}=\P^{\pm},\quad\check{\Gamma}^{n}= \left\lbrace
\begin{array}{lll}
\P^{\pm}&\mbox{if $n$ is odd},&
\\
\C^{\times}&\mbox{if $n$ is even}.&
\end{array}
\right.
\end{eqnarray*}
\end{lemma}
\begin{proof}
Let us prove $\check{\Gamma}^{0}\subseteq\P^{\pm}$. Suppose  $\widehat{T}\C^{0}T^{-1}\subseteq\C^{0}$; then $\widehat{T}T^{-1}=\alpha e,$ $\alpha\in\F^{\times}$. Multiplying both sides of this equation on the right by $T$, we obtain $\widehat{T}=\alpha T$. Suppose $T=T_0 + T_1$, where $T_0\in\C^{(0)}$, $T_1\in\C^{(1)}$. We get $T_0 - T_1=\alpha T_0 + \alpha T_1$, i.e. $T_0=\alpha T_0$, $-T_1=\alpha T_1$; therefore, $\alpha=1,$ $T_1=0$ or $\alpha=-1,$ $T_0=0$. Thus we have $T\in\C^{(0)\times}\cup\C^{(1)\times}=\P^{\pm}$. Let us prove $\check{\Gamma}^{n}\subseteq\P^{\pm}$ in the case of odd $n$ in a similar way. Suppose $T\in\check{\Gamma}^{n}$. Since $\C^n\subset\Z$ in the case of odd $n$, we get $\widehat{T}e_{1\ldots n}T^{-1}=\widehat{T}T^{-1}e_{1\ldots n}=\alpha e_{1\ldots n}$, where $\alpha\in\F^{\times}$. Thus, $\widehat{T}T^{-1}=\alpha e$ and we obtain $T\in\P^{\pm}$. 

Let us prove $\P^{\pm}\subseteq\check{\Gamma}^{0}$ in the case of arbitrary natural $n\geq1$. Suppose $T\in\P^{\pm}$; then $\widehat{T}T^{-1}=\pm T T^{-1}=\pm e\in\C^{0}$. Hence, $T\in\check{\Gamma}^{0}$ for all $n\geq1$.

Let us prove $\P^{\pm}\subseteq\check{\Gamma}^{n}$ in the case of odd $n$. Suppose $T\in\P^{\pm}$; then $\widehat{T} e_{1\ldots n} T^{-1}=\pm T e_{1\ldots n} T^{-1}=\pm e_{1\ldots n} T T^{-1}=\pm e_{1\ldots n}$. Thus, $\widehat{T}\C^{n}T^{-1}\subseteq\C^{n}$ if $n$ is odd.

Let us prove $\check{\Gamma}^{n}=\C^{\times}$ in the case of even $n$. The left set is a subset of the right one. Let us prove $\C^{\times}\subseteq\check{\Gamma}^{n}$. Suppose $T\in\C^{\times}$. Since $e_{1\ldots n}$ commutes with all even elements and anticommutes with all odd elements in the case of even $n$, we have $\widehat{T}e_{1\ldots n}=e_{1\ldots n}T$. Therefore, $\widehat{T}e_{1\ldots n}T^{-1}=e_{1\ldots n}T T^{-1}=e_{1\ldots n}$ and $\widehat{T}\C^{n}T^{-1}\subseteq\C^{n}$. Thus, $T\in\check{\Gamma}^{n}$. 
\end{proof}

In the paper \cite{OnInner},
the groups $\Gamma^{0}$ and $\Gamma^{n}$ preserving the subspaces of grades $0$ and $n$ respectively under the adjoint representation are considered. 
Note that we have
\begin{eqnarray}\label{g0gnoninner}
\Gamma^{0}=\C^{\times},\quad\Gamma^{n}= \left\lbrace
\begin{array}{lll}
\C^{\times}&\mbox{if $n$ is odd},&
\\
\P^{\pm}&\mbox{if $n$ is even}.&
\end{array}
\right.
\end{eqnarray}

Let us consider the groups preserving the direct sum of the subspaces of grades $0$ and $n$ under the adjoint representation and the twisted adjoint representation  respectively
\begin{eqnarray*}
\Gamma^{0n}&:=&\{T\in\C^{\times}:\quad T\C^{0n}T^{-1}\subseteq\C^{0n} \},
\\
\check{\Gamma}^{0n}&:=&\{T\in\C^{\times}:\quad \widehat{T}\C^{0n}T^{-1}\subseteq\C^{0n} \}.
\end{eqnarray*}

\begin{lemma}\label{Gamma0n}
We have
\begin{eqnarray}
\Gamma^{0n}&=& \left\lbrace
\begin{array}{lll}
\C^{\times}&\mbox{if $n$ is odd},&
\\
\{T\in\C^{\times}:\quad\widehat{T}T^{-1}\in\C^{0n\times}\}=\P^{\pm}&\mbox{if $n$ is even},&
\end{array}\label{eqchg0n}
\right.
\\
\check{\Gamma}^{0n}&=&\{T\in\C^{\times}:\quad \widehat{T}T^{-1}\in\C^{0n\times}\}=\P.\label{eqchgamma0n}
\end{eqnarray}
\end{lemma}
\begin{proof} 
Let us prove $\Gamma^{0n}=\C^{\times}$ if $n$ is odd. The left set is a subset of the right one. The statement $\C^{\times}\subseteq\Gamma^{0n}$ follows from $\C^{\times}=\Gamma^{0}=\Gamma^{n}$ (\ref{g0gnoninner}). 

Consider the case of even $n$. Let us prove $\Gamma^{0n}=\{T\in\C^{\times}:\quad\widehat{T}T^{-1}\in\C^{0n\times}\}$. Note that the condition  $TUT^{-1}\in\C^{0n}$ with substituting $U=e$ is true for  any $T\in\C^{\times}$; therefore, it is enough to consider it with substituting $U=e_{1\ldots n}$. The considered sets are equal because we get $T e_{1\ldots n}T^{-1}=e_{1\ldots n}\widehat{T}T^{-1}\in\C^{0n\times}$ using that $e_{1\ldots n}$ commutes with all even elements and anticommutes with all odd elements.
Let us prove $\P^{\pm}\subseteq\Gamma^{0n}$. Suppose $T\in\P^{\pm}$; then $T\in\Gamma^{n}$, $T\in\Gamma^{0}$ (\ref{g0gnoninner}). Thus, $T\in\Gamma^{0n}$.
Now let us prove $\{T\in\C^{\times}:\quad\widehat{T}T^{-1}\in\C^{0n\times}\}\subseteq\P^{\pm}$. Suppose $\widehat{T}T^{-1}\in\C^{0n\times}$; then  $\widehat{T}=WT$, where $W\in\C^{0n\times}$. Suppose also $T=T_0+T_1$, where $T_0\in\C^{(0)}$, $T_1\in\C^{(1)}$. Then we have $T_0 - T_1=W(T_0 + T_1)$, i.e. $T_0=W T_0 $, $-T_1=W T_1 $. Hence, $(W-e)T_0=0$, $(W+e)T_1=0$. 
If at least one of the two elements $W-e$ and $W+e$ is invertible, then we obtain either $T_0=0$ or $T_1=0$, and $T\in\C^{(0)\times}\cup\C^{(1)\times}$. 
Assume both $W-e$ and $W+e$ are non invertible. For $W=\alpha e + \beta e_{1\ldots n}$, $\alpha,\beta\in\F$, we obtain $W\pm e=(\alpha \pm 1)e + \beta e_{1\ldots n}$. Let us note that
\begin{eqnarray*}
((\alpha \pm 1)e + \beta e_{1\ldots n})((\alpha \pm 1)e-\beta e_{1\ldots n})=(\alpha \pm 1)^2 e- \beta^2 (e_{1\ldots n})^2 .
\end{eqnarray*}
Both $W-e$ and $W+e$ are non invertible only if either $\alpha=0$, $\beta=\pm 1$, $(e_{1\ldots n})^2=e$ or $\alpha=0$, $\beta=\pm i$, $(e_{1\ldots n})^2=-e$ (in the case $\F=\CC$). Then we have $W=\pm e_{1\ldots n}$ or $W=\pm i e_{1\ldots n}$; therefore, 
\begin{eqnarray*}
\widehat{T}=WT=W(T_0 + T_1)= WT_0 + WT_1=T_0 W - T_1 W=\widehat{T}W,
\end{eqnarray*}
and we get a contradiction, since $\widehat{T}\in\C^{\times}$ and $W\neq e$.

Let us prove $\check{\Gamma}^{0n}=\{T\in\C^{\times}:\quad \widehat{T}T^{-1}\in\C^{0n\times}\}$.
The condition $\widehat{T}\C^{0}T^{-1}\subseteq\C^{0n}$ is equivalent to $\widehat{T}T^{-1}\in\C^{0n\times}$.
In the case of odd $n$, the condition $\widehat{T}\C^{n}T^{-1}\subseteq\C^{0n}$ is equivalent to  $\widehat{T}T^{-1}\in\C^{0n\times}$ as well because $e_{1\ldots n}\in\Z^{\times}$.  In the case of even $n$, we have $\widehat{T}e_{1\ldots n}T^{-1}=e_{1\ldots n}T T^{-1}=e_{1\ldots n}$, since $e_{1\ldots n}$ commutes with all even elements and anticommutes with all odd elements, so the condition $\widehat{T}\C^{n}T^{-1}\subseteq\C^{0n}$ holds automatically.

In the case of even $n$, the equality $\{T\in\C^{\times}:\quad \widehat{T}T^{-1}\in\C^{0n\times}\}=\P$ follows from (\ref{eqchg0n}). 

Let us prove $\P\subseteq\check{\Gamma}^{0n}$ in the case of odd $n$. 
We have $\P=\C^{0n\times}\C^{(0)\times}$. Suppose $T=WU\in\P$ for some $W\in\C^{0n\times}$, $U\in\C^{(0)\times}$; then we get
\begin{eqnarray*}
\widehat{T}HT^{-1}=H\widehat{T}T^{-1}=H\widehat{(WU)}(WU)^{-1}=H\widehat{W}U U^{-1}W^{-1}=H\widehat{W}W^{-1}\in\C^{0n}\qquad\forall H\in\C^{0n}
\end{eqnarray*} using $\C^{0n}=\Z$. Thus, $T\in\check{\Gamma}^{0n}$.

Let us prove $\{T\in\C^{\times}:\quad \widehat{T}T^{-1}\in\C^{0n\times}\}\subseteq\P$ in the case of odd $n$. Suppose $\widehat{T}T^{-1}=\alpha e +\beta e_{1\ldots n}\in\C^{0n\times}$, where $\alpha$, $\beta\in\F$. Suppose also $T=T_0+T_1$, where $T_0\in\C^{(0)}$, $T_1\in\C^{(1)}$. Then $\widehat{T}=T_0-T_1=(\alpha e +\beta e_{1\ldots n})(T_0+T_1)$; therefore, 
\begin{eqnarray}\label{eqt0t1}
T_0=\alpha T_0 + \beta e_{1\ldots n}T_1,\qquad -T_1=\alpha T_1 +\beta e_{1\ldots n}T_0.
\end{eqnarray} 
If $\beta=0$, then either $\alpha=1$, $T_1=0$ and $T\in\C^{(0)\times}$ or $\alpha=-1$, $T_0=0$ and $T\in\C^{(1)\times}$.
If $\beta\neq0$, then from the first equality (\ref{eqt0t1}) it follows that $T_1=\frac{1-\alpha}{\beta}(e_{1\ldots n})^{-1}T_0.$
Thus, $T=T_0+T_1=T_0+\frac{1-\alpha}{\beta} (e_{1\ldots n})^{-1} T_0=(e+\frac{1-\alpha}{\beta} (e_{1\ldots n})^{-1})T_0$, i.e. $T\in\Z^{\times}\C^{(0)\times}=\Z^{\times}(\C^{(0)\times}\cup\C^{(1)\times})=\P$, since the element $e+\frac{1-\alpha}{\beta} (e_{1\ldots n})^{-1}$ is invertible. 
Assume $e+\frac{1-\alpha}{\beta}(e_{1\ldots n})^{-1}$ is non invertible; then
\begin{eqnarray*}
(e+\frac{1-\alpha}{\beta}(e_{1\ldots n})^{-1})(e-\frac{1-\alpha}{\beta}(e_{1\ldots n})^{-1})=e - \frac{(1-\alpha)^2}{\beta^2}((e_{1\ldots n})^{-1})^2=0,
\end{eqnarray*}
i.e. $\beta^2(e_{1\ldots n})^2=(1-\alpha)^2 e$ and  $\alpha\neq1$. 
Using this equality and multiplying the first equality (\ref{eqt0t1}) by $\beta e_{1\ldots n}$, we get $(1-\alpha)\beta e_{1\ldots n}T_0=(1-\alpha)^2 T_1$, i.e. $\beta e_{1\ldots n}T_0=(1-\alpha)T_1.$
Substituting this expression into the second equality (\ref{eqt0t1}), we obtain $-T_1=\alpha T_1 + (1-\alpha)T_1$. Thus, $T_1=0$. Using (\ref{eqt0t1}), we get $T=0$ and a contradiction.
\end{proof}

We need the following lemma to prove Theorems \ref{theoremP*G0'G1'} and \ref{maintheoremQ*}. The statement (\ref{1.31}) is proved in the paper \cite{OnInner}, the statement (\ref{3.31}) is new. We give the proof of both statements for the reader's convenience.

\begin{lemma}\label{an1.3lemma}
We have
\begin{eqnarray}
&&\{X\in\C:\quad XV=VX,\quad \forall V\in\C^{(0)}\}=\C^0\oplus\C^n\qquad\forall n;\label{1.31}
\\
&&\{X\in\C:\quad \widehat{X}V=VX,\quad \forall V\in\C^{(0)}\}=
\left\lbrace
\begin{array}{lll}
\C^0\oplus\C^n&\mbox{if $n$ is even},&\label{2.31}
\\
\C^0&\mbox{if $n$ is odd}.&\label{3.31}
\end{array}
\right.
\end{eqnarray}
\end{lemma}
\begin{proof}
Let us prove that the left set is a subset of the right one in (\ref{1.31}). Suppose $V=e_{ab}$, $\forall a<b$; then $Xe_{ab}=e_{ab}X$. Consider the case $a=1$, $b=2$. Let us represent $X$ in the form $X=A + e_1 B + e_2 C + e_{12} D,$ where the elements $A,B,C,D\in\C$ contain neither $e_1$ nor $e_2$. We have $(A + e_1 B+e_2 C + e_{12}D)e_{12}=e_{12}(A + e_1 B+e_2 C + e_{12}D)$. Since $A e_{12}=e_{12}A$, $D e_{12}=e_{12}D$, $e_1 B e_{12}=-e_{12}e_1 B$, $e_2 C e_{12}=-e_{12}e_2 C$, we obtain $2 e_1 B e_{12} + 2 e_2 C e_{12}=0$. Multiplying both sides of this equation on the right by $(e_{12})^{-1}$, we get $e_1 B=-e_2 C$. Since $B$ and $C$ contain neither $e_1$ nor $e_2$, we obtain $B=C=0$. Acting similarly to all other $a<b$, we get $X=\alpha e + \beta e_{1\ldots n}$, where $\alpha,\beta\in\F$.

Let us prove that the right set is a subset of the left one in (\ref{1.31}). Suppose $X\in\C^{0}\oplus\C^{n}$. If $n$ is odd, then we have $X\in\Z$, so $XV=VX$ for any $V\in\C^{(0)}$. If $n$ is even, then $e_{1\ldots n}$ commutes with all even elements, and we similarly obtain $XV=VX$ for any $V\in\C^{(0)}$.

Let us prove that in (\ref{3.31}), the left set is a subset of $\C^{0}\oplus\C^{n}$ if $n$ is even and $\C^0$ if $n$ is odd. 
Substituting $V=e\in\C^{(0)}$ into $\widehat{X}V=VX$, we get $\widehat{X}=X$ and $X\in\C^{(0)}$.
Therefore, we have $XV=VX$ for any $V\in\C^{(0)}$. Using (\ref{1.31}), we obtain $X\in\C^{0}\oplus\C^{n}$ in the case of arbitrary $n$. Since $X\in\C^{(0)}$, we have $X\in\C^{0}$ if $n$ is odd,  $X\in\C^{0}\oplus\C^{n}$ if $n$ is even. 

Let us prove that in the case of even $n$, the set $\C^{0}\oplus\C^{n}$ is a subset of the left set in (\ref{3.31}). Suppose $X\in\C^{0}\oplus\C^{n}\subseteq\C^{(0)}$. Then we have $\widehat{X}V=VX$ for any $V\in\C^{(0)}$, since $e_{1\ldots n}$ commutes with all even elements and $\widehat{X}=X$. 

Let us prove that in the case of odd $n$, the set $\C^{0}$ is a subset of the left set in (\ref{3.31}). Suppose $X\in\C^{0}$. Since $\widehat{X}=X$ and $X\in\Z$, we obtain $\widehat{X}V=VX$ for any $V\in\C^{(0)}$.
\end{proof}

Let us use the following notation for the groups preserving the even subspace and the odd subspace under the adjoint representation:
\begin{eqnarray*}
\Gamma^{(k)}&:=&\{T\in\C^{\times}:\quad T\C^{(k)}T^{-1}\subseteq\C^{(k)} \},\quad k=0,1;
\end{eqnarray*}
and the twisted adjoint representation:
\begin{eqnarray*}
\check{\Gamma}^{(k)}&:=&\{T\in\C^{\times}:\quad \widehat{T}\C^{(k)}T^{-1}\subseteq\C^{(k)} \},\quad k=0,1.
\end{eqnarray*}

\begin{theorem}\label{theoremP*G0'G1'}
The following groups coincide
\begin{eqnarray*}
\P^{\pm}=\check{\Gamma}^{(0)}=\check{\Gamma}^{(1)}.
\end{eqnarray*}
\end{theorem}
\begin{proof}
Let us prove $\P^{\pm}\subseteq\check{\Gamma}^{(0)}$, $\P^{\pm}\subseteq\check{\Gamma}^{(1)}$. If $T\in\C^{(0)\times}$, then $\widehat{T}=T\in\C^{(0)}$, $T^{-1}\in\C^{(0)}$. Similarly, if $T\in\C^{(1)\times}$, then $\widehat{T}=-T\in\C^{(1)}$, $T^{-1}\in\C^{(1)}$. Therefore, in both cases we have $\widehat{T}\C^{(0)}T^{-1}\subseteq\C^{(0)}$, $\widehat{T}\C^{(1)}T^{-1}\subseteq\C^{(1)}$ by (\ref{even}). 

 Let us prove $\check{\Gamma}^{(1)}\subseteq\P^{\pm}$. Suppose $T\in\C^{\times}$ satisfies  $\widehat{T}\C^{(1)}T^{-1}\subseteq\C^{(1)}$; then we obtain
\begin{eqnarray}\label{TuT(0)(1)}
\widehat{T}U T^{-1}=-(\widehat{T}U T^{-1}){\widehat{\;\;}}=-T\widehat{U}\widehat{T^{-1}}=T U \widehat{T^{-1}}\qquad \forall U\in\C^{(1)}
\end{eqnarray}
using the property of the grade involution (\ref{^uv=^u^v}).
Multiplying both sides of the equation (\ref{TuT(0)(1)}) on the left by $\widehat{T^{-1}}$, on the right by $T$, we get $U =\widehat{T^{-1}}T U\widehat{T^{-1}} T$, which is equivalent to
\begin{eqnarray*}
\widehat{(T^{-1}\widehat{T})} U (T^{-1}\widehat{T})^{-1} =U \qquad \forall U\in\C^{(1)};
\end{eqnarray*}
therefore,
\begin{eqnarray}\label{adT-1^T(1)}
\check{\ad}_{T^{-1}\widehat{T}}(U)=U\qquad\forall U\in\C^{(1)}.
\end{eqnarray}
In particular, (\ref{adT-1^T(1)}) is true for any generator $U=e_a\in\C^{(1)}$, $a=1,\ldots,n$. Thus, $T^{-1}\widehat{T}\in\ker(\check{\ad})=\C^{0\times}$ (\ref{kerchad}). We multiply both sides of this equation on the left by $T$, on the right by $T^{-1}$ and obtain $\widehat{T}T^{-1}\in\C^{0\times}$; hence, $T\in\check{\Gamma}^{0}$. Since $\check{\Gamma}^{0}=\P^{\pm}$ by Lemma \ref{lemmag0'gn'}, we have $T\in\P^{\pm}$.

Let us prove $\check{\Gamma}^{(0)}\subseteq\P^{\pm}$. Suppose $T\in\C^{\times}$ satisfies $\widehat{T}\C^{(0)} T^{-1}\subseteq\C^{(0)}$, then we obtain
\begin{eqnarray}\label{Tut(00)}
\widehat{T}U T^{-1}=(\widehat{T}U T^{-1}){\widehat{\;\;}}=T \widehat{U}\widehat{T^{-1}}=T U\widehat{T^{-1}}\qquad \forall U\in\C^{(0)}
\end{eqnarray}
using the property of the grade involution (\ref{^uv=^u^v}).
Multiplying both sides of the equation (\ref{Tut(00)}) on the left by $T^{-1}$, on the right by $T$, we get $T^{-1}\widehat{T}U=U\widehat{T^{-1}}T$, which is equivalent to 
\begin{eqnarray*}
\widehat{(\widehat{T^{-1}}T)}U=U(\widehat{T^{-1}}T)\qquad\forall U\in\C^{(0)}.
\end{eqnarray*}
In the case of odd $n$, we get $\widehat{T^{-1}}T\in\C^{0}$ by Lemma \ref{an1.3lemma}; hence, $T^{-1}\widehat{T}\in\C^{0}$. Therefore, we obtain $T\in\check{\Gamma}^0=\P^{\pm}$ using Lemma \ref{lemmag0'gn'}.
In the case of even $n$, we get $\widehat{T^{-1}}T=\alpha e +\beta e_{1\ldots n}\in\C^{0}\oplus\C^{n}$, where $\alpha,\beta\in\F$, using Lemma \ref{an1.3lemma}.
We multiply both sides of this equation on the left by $\widehat{T}$, on the right by $T^{-1}$ and obtain $\widehat{T}(\alpha e +\beta e_{1\ldots n})T^{-1}=e$. Since $e_{1\ldots n}$ commutes with all even elements and anticommutes with all odd elements if $n$ is even, $\alpha \widehat{T}T^{-1}+\beta e_{1\ldots n}=e$. Thus, $\widehat{T}T^{-1}\in\C^{0n\times}$ and $T\in\check{\Gamma}^{0n}=\P^{\pm}$ (Lemma \ref{Gamma0n}).
\end{proof}

The group $\P$ (\ref{P}) is related in a similar way to the groups that preserve the even subspace and the odd subspace under the adjoint representation:
\begin{eqnarray*}
\P=\Gamma^{(0)}=\Gamma^{(1)}.
\end{eqnarray*}

Let us note that the adjoint representation coincides up to a sign with the twisted adjoint representation when applying it to any element $T$ of the group $\P^{\pm}$ (or of its subgroups $\Q^{\pm}$, $\Q'$, $\check{\Gamma}^k$, $k=1, \ldots, n$, considered further in this paper):
\begin{eqnarray*}
\check{\ad}_T(U)=\pm\ad_T(U)\qquad\forall T\in\P^{\pm},\quad\forall U\in\C.
\end{eqnarray*}

\section{The groups $\A_{\pm}$, $\check{\Gamma}^{\overline{12}}$, and $\check{\Gamma}^{\overline{03}}$}\label{grA*}

 Let us use the following notation for the groups preserving the subspaces of quaternion types or their direct sums under the adjoint representation:
\begin{eqnarray*}
&&\Gamma^{\overline{k}}:=\{T\in\C^{\times}:\quad T\C^{\overline{k}}T^{-1}\subseteq\C^{\overline{k}} \},\quad k=0,1,2,3;
\\
&&\Gamma^{\overline{kl}}:=\{T\in\C^{\times}:\quad T\C^{\overline{kl}}T^{-1}\subseteq\C^{\overline{kl}} \},\quad k,l=0,1,2,3;
\end{eqnarray*}
and the twisted adjoint representation:
\begin{eqnarray*}
&&\check{\Gamma}^{\overline{k}}:=\{T\in\C^{\times}:\quad \widehat{T}\C^{\overline{k}}T^{-1}\subseteq\C^{\overline{k}} \},\quad k=0,1,2,3;
\\
&&\check{\Gamma}^{\overline{kl}}:=\{T\in\C^{\times}:\quad \widehat{T}\C^{\overline{kl}}T^{-1}\subseteq\C^{\overline{kl}} \},\quad k,l=0,1,2,3.
\end{eqnarray*}

We need the following lemma in order to prove the main statements in Sections \ref{grA*}--\ref{sectiongk'}.

\begin{lemma}\label{modd^XU=UX}
Consider an arbitrary element $X\in\C$ and an arbitrary fixed subset $\H$ of the set $\C^{(0)}\cup\C^{(1)}$. We have
\begin{eqnarray*}
\widehat{X}U=U X\qquad\forall U\in \H\qquad\Rightarrow\qquad \widehat{X}(U_1 \cdots U_m)= (U_1 \cdots U_m) X\qquad \forall U_1,\ldots,U_m\in \H
\end{eqnarray*}
for any odd natural number $m$.
\end{lemma}
\begin{proof}
We prove this lemma by induction on odd $m$. For $m=1$, there is nothing to prove. Assume we have 
\begin{eqnarray*}
\widehat{X}(U_1 \cdots U_m)= (U_1 \cdots U_m) X\qquad \forall U_1,\ldots,U_m\in\H
\end{eqnarray*}
for some odd natural $m$. We obtain
\begin{eqnarray*}
\widehat{X}U=UX\quad\Leftrightarrow\quad\widehat{\widehat{X}U}=\widehat{UX}\quad\Leftrightarrow\quad X\widehat{U}=\widehat{U}\widehat{X}\quad\Leftrightarrow\quad XU=U\widehat{X}\qquad\forall U\in\H
\end{eqnarray*}
using the property of the grade involution (\ref{^uv=^u^v}). Then we get 
\begin{eqnarray*}
\widehat{X}(U_1 \cdots U_l U_{m+1} U_{m+2}) = (U_1 \cdots U_l) X U_{m+1} U_{m+2} = (U_1 \cdots U_{m+1}) \widehat{X} U_{m+2} = (U_1 \cdots U_{m+2}) X
\end{eqnarray*}
for $m+2$. This completes the proof of the lemma.
\end{proof}

Let us remark that the similar statement
\begin{eqnarray*}
X U=U X\quad\forall U\in\H\qquad\Rightarrow\qquad X(U_1 \cdots U_m)= (U_1\cdots U_m) X\quad \forall U_1,\ldots,U_m\in\H
\end{eqnarray*}
is true for any natural $m$ and any subset $\H$ of the set $\C$. 
\\

Let us consider the following two norm functions \cite{lg1, lounesto, ABS, lg2, spin1, spin2}
\begin{eqnarray}\label{normfunc}
\psi(T):=\widetilde{T}T,\qquad \chi(T):=\widehat{\widetilde{T}}T,\qquad\forall T\in\C,
\end{eqnarray}
which are widely used in the theory of spin groups.
It is shown \cite{OnInner} that
\begin{eqnarray}\label{nfunc}
\widetilde{T}T\in\C^{\overline{01}},\qquad \widehat{\widetilde{T}}T\in\C^{\overline{03}},\qquad\forall T\in\C.
\end{eqnarray}
Let us consider the group
\begin{eqnarray}
\A_{\pm}&:=&\{T\in\C^\times:\quad\widetilde{T}T\in\C^{0\times}\}.\label{A*}
\end{eqnarray}
Note that the group $\A_{\pm}$ coincides with the group $\A$  in the particular cases:
\begin{eqnarray*}
\A_{\pm}=\A\qquad n=0,2,3\mod{4}.
\end{eqnarray*}
Let us note that, since $\A=\P$ in the cases $n\leq3$ (see \cite{OnInner}), we have
\begin{eqnarray}\label{a*ap}
\P^{\pm}=\P=\A_{\pm}=\A\qquad n=0,2;\qquad \P^{\pm}\subset\P=\A_{\pm}=\A\qquad n=3.
\end{eqnarray}
Also from the next theorem it follows that
\begin{eqnarray*}
\A_{\pm}=\P^{\pm}\qquad n=1.
\end{eqnarray*}

\begin{theorem}\label{Acoincide}
The following groups coincide
\begin{eqnarray*}
\A_{\pm}=\check{\Gamma}^{\overline{12}}=\check{\Gamma}^{\overline{03}}.
\end{eqnarray*}
\end{theorem}
\begin{proof}
Let us prove $\A_{\pm}\subseteq\check{\Gamma}^{\overline{12}}$. Suppose $\widetilde{T}T=\lambda e,$ $\lambda\in\F^{\times}$. Then we have
\begin{eqnarray*}
(\widehat{T}U T^{-1}){\widehat{\widetilde{\;\;}}}=-\widehat{\widetilde{T^{-1}}}U \widetilde{T}=-\frac{1}{\lambda}\widehat{T}U \lambda T^{-1}=-\widehat{T}U T^{-1}\qquad \forall U\in\C^{\overline{12}}
\end{eqnarray*}
using the property of the Clifford conjugation (\ref{^uv=^u^v}). The Clifford conjugation changes the sign of the element $\widehat{T}U T^{-1}$; therefore, $\widehat{T}U T^{-1}\in\C^{\overline{12}}$ for any $U\in\C^{\overline{12}}$. Thus, $T\in\check{\Gamma}^{\overline{12}}$.

Let us prove $\A_{\pm}\subseteq\check{\Gamma}^{\overline{03}}$. We obtain
\begin{eqnarray*}
(\widehat{T}U T^{-1}){\widehat{\widetilde{\;\;}}}=\widehat{\widetilde{T^{-1}}}U \widetilde{T}=\frac{1}{\lambda}\widehat{T}U \lambda T^{-1}=\widehat{T}U T^{-1}\qquad \forall U\in\C^{\overline{03}}.
\end{eqnarray*}
Since the Clifford conjugation does not change the element $\widehat{T}U T^{-1}$, we have $\widehat{T}U T^{-1}\in\C^{\overline{03}}$ for any $U\in\C^{\overline{03}}$; therefore, $T\in\check{\Gamma}^{\overline{03}}$.

Let us prove $\check{\Gamma}^{\overline{12}}\subseteq\A_{\pm}$.
Suppose $T\in\C^{\times}$ and $\widehat{T}\C^{\overline{12}}T^{-1}\subseteq\C^{\overline{12}}$; then 
\begin{eqnarray}\label{TuTov12}
\widehat{T}U T^{-1}=-(\widehat{T}U T^{-1}){\widehat{\widetilde{\;\;}}}=\widehat{\widetilde{T^{-1}}} U\widetilde{T}\qquad \forall U\in\C^{\overline{12}}.
\end{eqnarray}
Multiplying both sides of the equation (\ref{TuTov12}) on the left by $\widehat{\widetilde{T}}$, on the right by $\widetilde{T^{-1}}$, we get $\widehat{\widetilde{T}}\widehat{T}UT^{-1}\widetilde{T^{-1}}=U$ for any $U\in\C^{\overline{12}}$, which is equivalent to  $\widehat{(\widetilde{T}T)}U(\widetilde{T}T)^{-1}=U$ for any $U\in\C^{\overline{12}}$. Hence,
\begin{eqnarray}\label{ad_TT12}
\check{\ad}_{\widetilde{T}T}(U)=U\qquad\forall U\in\C^{\overline{12}}.
\end{eqnarray}
In particular, 
(\ref{ad_TT12}) is true for all generators $U=e_a\in\C^{\overline{12}}$, $a=1,\ldots,n$. Thus, $\widetilde{T}T\in\ker(\check{\ad})=\C^{0\times}$ (\ref{kerchad}).

Let us prove $\check{\Gamma}^{\overline{03}}\subseteq\A_{\pm}$ in the cases $n\geq4$. Suppose $T\in\C^{\times}$ satisfies $\widehat{T}\C^{\overline{03}} T^{-1}\subseteq\C^{\overline{03}}$; then
\begin{eqnarray}\label{TuTov03}
\widehat{T}U T^{-1}=(\widehat{T}U T^{-1}){\widehat{\widetilde{\;\;}}}=\widehat{\widetilde{T^{-1}}} U\widetilde{T}\qquad \forall U\in\C^{\overline{03}}.
\end{eqnarray}
We multiply both sides of the equation (\ref{TuTov03}) on the left by $\widehat{\widetilde{T}}$, on the right by $\widetilde{T^{-1}}$ and obtain
\begin{eqnarray}\label{ad_TT03}
\check{\ad}_{\widetilde{T}T}(U)=U\qquad\forall U\in\C^{\overline{03}}.
\end{eqnarray}
Since $n\geq4$, any generator can always be represented as the product of three elements of grade $3$. For instance, we can get the generator $e_2$ by the multiplication of $e_{123}$, $e_{124}$, $e_{234}$. Now, using Lemma \ref{modd^XU=UX}, we obtain that any generator $U=e_a$, $a=1,\ldots,n$, satisfies (\ref{ad_TT03}). Thus, $\widetilde{T}T\in\ker(\check{\ad})=\C^{0\times}$ (\ref{kerchad}) and the proof is completed.

Let us prove $\check{\Gamma}^{\overline{03}}\subseteq\A_{\pm}$ in the cases $n\leq3$. 
In the cases $n=1, 2$, we have $\check{\Gamma}^{\overline{03}}=\check{\Gamma}^0=\P^{\pm}$, and we need to prove $\P^{\pm}\subseteq\A_{\pm}$.
Suppose $T\in\P^{\pm}=\C^{(0)\times}\cup\C^{(1)\times}$; hence, $\widetilde{T}T\in\C^{(0)\times}$ by (\ref{even}).
Since $\widetilde{T}T\in\C^{\overline{01}\times}$ (\ref{nfunc}), we obtain $\widetilde{T}T\in\C^{(0)\times}\cap\C^{\overline{01}\times}=\C^{0\times}$ and $T\in\A_{\pm}$. 
Let us prove $\check{\Gamma}^{\overline{03}}\subseteq\A_{\pm}$ if $n=3$. Suppose $T\in\check{\Gamma}^{\overline{03}}$. Substituting $U=e$ into the condition $\widehat{T}UT^{-1}\in\C^{\overline{03}}$ for any $U\in\C^{\overline{03}}$, we obtain
\begin{eqnarray}\label{TTZ}
\widehat{T}T^{-1}=W\in\Z^{\times}.
\end{eqnarray}
We multiply both sides of the equation (\ref{TTZ}) on the left by $\widetilde{T}$, on the right by $T W^{-1}$ and obtain $\widetilde{T}\widehat{T}W^{-1}=\widetilde{T}W T W^{-1}$, which is equivalent to 
\begin{eqnarray}\label{ttwtt}
\widehat{(\widehat{\widetilde{T}}T)} W^{-1}=\widetilde{T}T.
\end{eqnarray}
Since $W^{-1}\in\C^{\overline{03}\times}=\Z^{\times}$ and $\widehat{\widetilde{T}}T\in\C^{\overline{03}\times}=\Z^{\times}$ (\ref{nfunc}), the left part of (\ref{ttwtt}) is in $\C^{\overline{03}\times}$. Since $\widetilde{T}T\in\C^{\overline{01}\times}$ (\ref{nfunc}), the right part of (\ref{ttwtt}) is in $\C^{\overline{01}\times}$. Therefore, $\widetilde{T}T\in\C^{\overline{03}\times}\cap\C^{\overline{01}\times}=\C^{0\times}$ and $T\in\A_{\pm}$.
\end{proof}

In the paper \cite{OnInner}, the group $\B=\{T\in\C^\times:\quad\widehat{\widetilde{T}}T\in\Z^{\times}\}$ (\ref{B}) is considered.
Let us remark that the group $\A_{\pm}$ preserves the direct sum of the subspaces of the same quaternion types under $\check{\ad}$ as the group $\B$ preserves under $\ad$ (Theorem 5.2 in \cite{OnInner}):
\begin{eqnarray*}
\B=\Gamma^{\overline{12}}=\Gamma^{\overline{03}}.
\end{eqnarray*}

\section{The groups $\B_{\pm}$, $\check{\Gamma}^{\overline{01}}$, and $\check{\Gamma}^{\overline{23}}$}\label{grB*}

Let us consider the group
\begin{eqnarray}\label{B*}
\B_{\pm}&:=&\{T\in\C^\times:\quad\widehat{\widetilde{T}}T\in\C^{0\times}\}.
\end{eqnarray}
Note that it follows from (\ref{B*}) and (\ref{nfunc}) that
\begin{eqnarray*}
\B_{\pm}=\B\qquad n=0,1,2\mod{4},
\end{eqnarray*}
and then from the paper \cite{OnInner} it follows that
\begin{eqnarray}\label{bbc}
\B_{\pm}=\B=\C^{\times}\qquad n=1,2.
\end{eqnarray}

\begin{lemma}\label{pab23}
We have
\begin{eqnarray}
&&\P^{\pm}=\A_{\pm}\neq\B_{\pm}=\C^{\times}\qquad n=1,2;
\\
&&\P^{\pm}\neq\A_{\pm},\qquad \P^{\pm}\neq\B_{\pm}, \qquad \A_{\pm}\neq\B_{\pm}\qquad n\geq3.
\end{eqnarray}
\end{lemma}
\begin{proof}
Let us consider the cases $n=1,2$. We have $\P^{\pm}=\A_{\pm}$ by Theorem \ref{Acoincide} and $\B_{\pm}=\C^{\times}$ by (\ref{bbc}).

Let us prove $\P^{\pm}\neq\B_{\pm}$ in the case of arbitrary $n\geq1$.
Consider the element
\begin{eqnarray*}
S = e + 2e_1\not\in\C^{(0)}\cup\C^{(1)},\qquad S\in\C^{\times}.
\end{eqnarray*}
We have
\begin{eqnarray*}
\widehat{\widetilde{S}}S=(e-2 e_1)(e+2e_1)=e-4(e_1)^2\in\C^{0\times};
\end{eqnarray*}
thus, $S\in\B_{\pm}$, $S\not\in\P^{\pm}$. 

Let us prove $\A_{\pm}\neq\P^{\pm}$ and $\A_{\pm}\neq\B_{\pm}$ in the cases $n\geq3$.
Consider the element
\begin{eqnarray*}
T=e_1 + 3e_{23}\not\in\C^{(0)}\cup\C^{(1)},\qquad T\in\C^{\times}.
\end{eqnarray*}
We have
\begin{eqnarray*}
&&\widetilde{T}T=(e_1 - 3 e_{23})(e_1 + 3 e_{23})=(e_1)^2 - 9 (e_{23})^2\in\C^{0\times},
\\
&&\widehat{\widetilde{T}}T=(-e_1 - 3 e_{23})(e_1 + 3e_{23})=-(e_1)^{2} - 9(e_{23})^{2} - 6 e_{123}\not\in\C^{0\times};
\end{eqnarray*}
therefore, $T\in\A_{\pm}$, $T\not\in\P^{\pm}$, and $T\not\in\B_{\pm}$. 
\end{proof}

\begin{theorem}\label{Bcoincide}
The following groups coincide
\begin{eqnarray*}
\B_{\pm}=\check{\Gamma}^{\overline{01}}=\check{\Gamma}^{\overline{23}}.
\end{eqnarray*}
\end{theorem}
\begin{proof}
Let us prove $\B_{\pm}\subseteq\check{\Gamma}^{\overline{01}}$. Suppose $\widehat{\widetilde{T}}T=\lambda e,$ $\lambda\in\F^{\times}$. We get
\begin{eqnarray*}
(\widehat{T}U T^{-1}){\widetilde{\;\;}}=\widetilde{T^{-1}}\widetilde{U}\widehat{\widetilde{T}}=\frac{1}{\lambda}\widehat{T}U\lambda T^{-1}=\widehat{T}U T^{-1}\qquad \forall U\in\C^{\overline{01}}
\end{eqnarray*}
using the property of the reversion (\ref{^uv=^u^v}). The reversion does not change the element $\widehat{T}U T^{-1}$, so $\widehat{T}U T^{-1}\in\C^{\overline{01}}$ for any $U\in\C^{\overline{01}}$. Thus, $T\in\check{\Gamma}^{\overline{01}}$.

Let us prove $\B_{\pm}\subseteq\check{\Gamma}^{\overline{23}}$. We obtain
\begin{eqnarray*}
(\widehat{T}U T^{-1}){\widetilde{\;\;}}=\widetilde{T^{-1}}\widetilde{U}\widehat{\widetilde{T}}=-\frac{1}{\lambda}\widehat{T}U\lambda T^{-1}=-\widehat{T}U T^{-1}\qquad \forall U\in\C^{\overline{23}}.
\end{eqnarray*}
Since the reversion changes the sign of the element $\widehat{T}U T^{-1}$, we have $\widehat{T}U T^{-1}\in\C^{\overline{23}}$ for any $U\in\C^{\overline{23}}$. Thus, $T\in\check{\Gamma}^{\overline{23}}$.

Let us prove $\check{\Gamma}^{\overline{01}}\subseteq\B_{\pm}$. Suppose $T\in\C^{\times}$ satisfies $\widehat{T}\C^{\overline{01}}T^{-1}\subseteq\C^{\overline{01}}$; then
\begin{eqnarray}\label{Tut01}
\widehat{T}U T^{-1}=(\widehat{T}U T^{-1}){\widetilde{\;\;}}=\widetilde{T^{-1}} U \widehat{\widetilde{T}}\qquad \forall U\in\C^{\overline{01}}.
\end{eqnarray}
Multiplying both sides of the equation (\ref{Tut01}) on the left by $\widetilde{T}$, on the right by $\widehat{\widetilde{T^{-1}}}$, we obtain $\widetilde{T}\widehat{T}U T^{-1}\widehat{\widetilde{T^{-1}}}= U$ for any $U\in\C^{\overline{01}}$, which is equivalent to the equation $\widehat{(\widehat{\widetilde{T}}T)}U (\widehat{\widetilde{T}}T)^{-1}=U$ for any $U\in\C^{\overline{01}}$; therefore, 
\begin{eqnarray}\label{ad^_TT01}
\check{\ad}_{\widehat{\widetilde{T}}T}(U)=U\qquad\forall U\in\C^{\overline{01}}.
\end{eqnarray}
In particular,  (\ref{ad^_TT01}) is true for any generator $U=e_a\in\C^{\overline{01}}$, $a=1,\ldots,n$. Thus, $\widehat{\widetilde{T}}T\in\ker(\check{\ad})=\C^{0\times}$ (\ref{kerchad}) and the proof is completed.

Let us prove $\check{\Gamma}^{\overline{23}}\subseteq\B_{\pm}$ in the cases $n\geq3$. Suppose $T\in\C^{\times}$ satisfies $\widehat{T}\C^{\overline{23}}T^{-1}\subseteq\C^{\overline{23}}$; then
\begin{eqnarray}\label{tut23}
\widehat{T}U T^{-1}=-(\widehat{T}U T^{-1}){\widetilde{\;\;}}=\widetilde{T^{-1}} U \widehat{\widetilde{T}}\qquad \forall U\in\C^{\overline{23}}.
\end{eqnarray}
We multiply both sides of the equation (\ref{tut23}) on the left by $\widetilde{T}$, on the right by $\widehat{\widetilde{T^{-1}}}$ and obtain
\begin{eqnarray}\label{ad^_TT23}
\check{\ad}_{\widehat{\widetilde{T}}T}(U)=U\qquad\forall U\in\C^{\overline{23}}.
\end{eqnarray}
Since $n\geq3$, any generator can always be represented as the product of two elements of grade $2$ and one element of grade $3$. For example, we can get the generator $e_1$ by the multiplication of $e_{12}$, $e_{13}$, $e_{123}$. Using Lemma \ref{modd^XU=UX}, we conclude that any generator $U=e_a$, $a=1,\ldots,n$, satisfies  (\ref{ad^_TT23}). Hence, $\widehat{\widetilde{T}}T\in\ker(\check{\ad})=\C^{0\times}$ (\ref{kerchad}) and the proof is completed.

In the cases $n\leq 2$, we get $\check{\Gamma}^{\overline{23}}\subseteq\C^{\times}=\B_{\pm}$ using (\ref{bbc}).
\end{proof}

In the paper \cite{OnInner}, the group $\A=\{T\in\C^\times:\quad\widetilde{T}T\in\Z^{\times}\}$ (\ref{A}) is considered.
Let us remark that the group $\B_{\pm}$ preserves the direct sum of the subspaces of the same quaternion types under $\check{\ad}$ as the group $\A$ preserves under $\ad$ (Theorem 4.2 in \cite{OnInner}):
\begin{eqnarray*}
\A=\Gamma^{\overline{01}}=\Gamma^{\overline{23}}.
\end{eqnarray*}

\section{The groups $\Q^{\pm}$ and $\check{\Gamma}^{\overline{\lowercase{k}}}$}\label{grQ*}

Let us consider the group
\begin{eqnarray}\label{Q*}
\Q^{\pm}&:=&\{T\in\C^{(0)\times}\cup\C^{(1)\times}:\quad \widetilde{T}T\in\C^{0\times}\}.
\end{eqnarray}
The groups $\Q^{\pm}$ and $\Q$ (\ref{Q}) are related in the following way:
\begin{eqnarray}\label{rel_Q}
\Q=\Z^{\times}\Q^{\pm},
\end{eqnarray}
in particular, these groups coincide in the case of even $n$:
\begin{eqnarray}\label{q=q}
\Q=\Q^{\pm},\qquad \mbox{$n$ is even}.
\end{eqnarray}

\begin{lemma}\label{lemmaQABP}
We have
\begin{eqnarray}\label{QABPcap}
\Q^{\pm}=\A_{\pm}\cap\P^{\pm}=\B_{\pm}\cap\P^{\pm}=\A_{\pm}\cap\B_{\pm},\qquad \Q^{\pm}\subseteq\P^{\pm},\qquad\Q^{\pm}\subseteq\A_{\pm},\qquad \Q^{\pm}\subset\B_{\pm}.
\end{eqnarray}
In the particular cases,
\begin{eqnarray}
&&\Q^{\pm}=\P^{\pm}=\A_{\pm}\neq\B_{\pm}=\C^{\times},\qquad n=1,2;\label{QABPcap2}
\\
&&\Q^{\pm}=\P^{\pm},\qquad \Q^{\pm}\neq\B_{\pm},\qquad\Q^{\pm}\neq\A_{\pm},\qquad n=3;\label{QABPcap3}
\\
&&\Q^{\pm}\neq\P^{\pm},\qquad n\geq4.\label{QABPcap4}
\end{eqnarray}
\end{lemma}
\begin{proof}
The first two equalities in (\ref{QABPcap}) follow from the definitions of the groups $\A_{\pm}$ (\ref{A*}), $\B_{\pm}$ (\ref{B*}), $\Q^{\pm}$ (\ref{Q*}). 
Let us prove $\A_{\pm}\cap\B_{\pm}=\Q^{\pm}$.
Suppose that some element $T\in\C^{\times}$ satisfies $\widetilde{T}T\in\C^{0\times}$, $\widehat{\widetilde{T}}T\in\C^{0\times}$. 
Then $\widehat{\widetilde{T}}T=\alpha \widetilde{T}T$, where $\alpha\in\F^{\times}$; therefore, we obtain $\widehat{T}T^{-1}\in\C^{0\times}$ and $T\in\check{\Gamma}^{0}=\P^{\pm}$ using Lemma \ref{lemmag0'gn'}. Thus, $\A_{\pm}\cap\B_{\pm}=\Q^{\pm}$.

In the cases $n=1,2$, we get $\Q^{\pm}=\P^{\pm}=\A_{\pm}$ using Lemma \ref{pab23} and the first equality in (\ref{QABPcap}).

In the case $n=3$, since $\P^{\pm}\subset\A_{\pm}$ (\ref{a*ap}), we obtain $\Q^{\pm}=\P^{\pm}$ by (\ref{QABPcap}).

Let us prove $\Q^{\pm}\neq\P^{\pm}$ if $n\geq4$. Consider the element $T=e+2e_{1234}\in\C^{(0)}$, which is invertible because $(e+2e_{1234})(e-2e_{1234})=e-4(e_{1234})^2\in\C^{0\times}.$
We have
\begin{eqnarray*}
\widetilde{T}T=(e+2e_{1234})(e+2e_{1234})=e +4e_{1234}+4(e_{1234})^2\not\in\C^{0\times},
\end{eqnarray*}
i.e. $T\in\P^{\pm}$, $T\not\in\A_{\pm}$, and $T\not\in\Q^{\pm}$.
\end{proof}

Let us consider the following groups introduced in the paper \cite{OnInner}:
\begin{eqnarray*}
\A'&=&\{T\in\C^{\times}:\quad \widetilde{T}T\in(\C^{0}\oplus\C^{n})^{\times}\},
\\
\B'&=&\{T\in\C^{\times}:\quad \widehat{\widetilde{T}}T\in(\C^{0}\oplus\C^{n})^{\times}\}.
\end{eqnarray*}
Note that the groups $\A'$, $\B'$ coincide with the groups $\A_{\pm}$, $\B_{\pm}$ in the particular cases by (\ref{nfunc}):
\begin{eqnarray}
\A'=\A_{\pm}\quad n=2,3\mod{4},&\qquad& \B'=\B_{\pm}\quad n=1,2\mod{4}.\label{aa*}
\end{eqnarray}
Note also that in some cases, the groups $\A'$, $\B'$ coincide with the groups $\A$ (\ref{A}), $\B$ (\ref{B}):
\begin{eqnarray}
\A'=\A,\quad\qquad \B'=\B,&\qquad& n=1,2,3\mod{4}.\label{bb*}
\end{eqnarray}
Let us consider the group $\Q'$ introduced in the paper \cite{OnInner}:
\begin{eqnarray}\label{Q'}
\Q'&=&\{T\in\Z^{\times}(\C^{(0)\times}\cup\C^{(1)\times}):\quad \widetilde{T}T\in(\C^{0}\oplus\C^{n})^{\times}\}.
\end{eqnarray}
We have (Lemma 6.2 in \cite{OnInner})
\begin{eqnarray}\label{a'b'q'}
\Q'=\A'\cap\B'.
\end{eqnarray}

\begin{lemma}\label{lemmaq*q'pp*}
We have
\begin{eqnarray}
\Q^{\pm}\subseteq\Q'\subseteq\P^{\pm},\qquad \mbox{$n$ is even}.\label{Q*Q'P*P}
\end{eqnarray}
In the particular cases, 
\begin{eqnarray}\label{Q'=P*=P}
\Q^{\pm}=\Q',\qquad n=2\mod{4};\qquad
\Q'=\P^{\pm},\qquad n=4.
\end{eqnarray}
\end{lemma}
\begin{proof}
The statements $\Q'\subseteq\P^{\pm}$ in the case of even $n$ and $\Q^{\pm}\subseteq\Q'$ in the case of arbitrary $n\geq1$ follow from the definitions of the groups $\Q'$ (\ref{Q'}), 
$\P^{\pm}$ (\ref{P}), and $\Q^{\pm}$ (\ref{Q}). 
Consider the case $n=4$. For $T\in\P^{\pm}=\C^{(0)\times}\cup\C^{(1)\times}$, we have $\widetilde{T}T\in\C^{(0)\times}$ by (\ref{even}); therefore,  $\widetilde{T}T\in\C^{\overline{01}\times}\cap\C^{(0)\times}=(\C^{0}\oplus\C^{4})^{\times}$. Thus, $\P^{\pm}\subseteq\Q'$.
 In the case  $n=2\mod{4}$, we get $\Q^{\pm}=\A_{\pm}\cap\B_{\pm}=\A'\cap\B'=\Q'$ using Lemma \ref{lemmaQABP}, (\ref{aa*}), and (\ref{a'b'q'}).
\end{proof}

\begin{theorem}\label{maintheoremQ*}
In the cases $n\geq3$, we have 
\begin{eqnarray*}
&&\Q^{\pm}=\check{\Gamma}^{\overline{1}}=\check{\Gamma}^{\overline{3}}=\check{\Gamma}^{\overline{2}}=\check{\Gamma}^{\overline{0}},\qquad n=1,2,3\mod{4},
\\
&&\Q^{\pm}=\check{\Gamma}^{\overline{1}}=\check{\Gamma}^{\overline{3}}\neq\Q^{'}=\check{\Gamma}^{\overline{2}}=\check{\Gamma}^{\overline{0}},\qquad n=0\mod{4}.
\end{eqnarray*}
In the exceptional cases, we have
\begin{eqnarray*}
\check{\Gamma}^{\overline{0}}=\check{\Gamma}^{\overline{1}}=\Q^{\pm}=\P^{\pm}\neq\check{\Gamma}^{\overline{2}}=\check{\Gamma}^{\overline{3}}=\C^{\times},\qquad n=1,2.
\end{eqnarray*}
As a consequence,
\begin{eqnarray*}
&\check{\Gamma}^{\overline{1}}\subseteq\check{\Gamma}^{\overline{k}},\quad k=0,1,2,3,
\qquad\bigcap_{k=0}^{3}\check{\Gamma}^{\overline{k}}=\check{\Gamma}^{\overline{1}},
\\
&\Q^{\pm}=\check{\Gamma}^{\overline{1}}=\{T\in\C^{\times}:\quad \widehat{T}\C^{\overline{k}}T^{-1}\subseteq\C^{\overline{k}},\quad k=0,1,2,3\}.
\end{eqnarray*}
\end{theorem}
\begin{proof}
Let us prove $\Q^{\pm}\subseteq\check{\Gamma}^{\overline{k}}$, $k=0,1,2,3$. Suppose $T\in\C^{(0)\times}\cup\C^{(1)\times}$ and $\widetilde{T}T=\lambda e$, $\lambda\in\F^{\times}$.
For an arbitrary element $U_{\overline{k}}\in\C^{\overline{k}}$, $k=0,1,2,3$, we get
\begin{eqnarray*}
(\widehat{T} U_{\overline{k}} T^{-1}){\widetilde{\;\;}}&=&\widetilde{T^{-1}}\widetilde{U_{\overline{k}}}\widehat{\widetilde{T}}=\frac{1}{\lambda}T\widetilde{U_{\overline{k}}}\lambda \widehat{T^{-1}}=\widehat{T}\widetilde{U_{\overline{k}}}T^{-1},
\\
(\widehat{T} U_{\overline{k}} T^{-1}){\widehat{\;\;}}&=&T\widehat{U_{\overline{k}}}\widehat{T^{-1}}=\widehat{T}\widehat{U_{\overline{k}}}T^{-1},
\end{eqnarray*}
where we use the properties of the reversion and the grade involution (\ref{^uv=^u^v}). 
Using the definition of quaternion types (\ref{qtdef}), we obtain $\widehat{T}U_{\overline{k}}T^{-1}\in\C^{\overline{k}}$ for any $U\in\C^{\overline{k}}$. Thus, $\Q^{\pm}\subseteq\check{\Gamma}^{\overline{k}}$, $k=0,1,2,3$.

Let us prove $\Q'\subseteq\check{\Gamma}^{\overline{0}}$, $\Q'\subseteq\check{\Gamma}^{\overline{2}}$ if $n=0\mod{4}$. Suppose $T\in\C^{(0)\times}\cup\C^{(1)\times}$ and $\widetilde{T}T=W\in(\C^{0}\oplus\C^{n})^{\times}$. For an arbitrary element $U_{\overline{k}}\in\C^{\overline{k}}$, we get
\begin{eqnarray*}
(\widehat{T}U_{\overline{k}}T^{-1}){\widetilde{\;\;}}&=&\widetilde{T^{-1}}\widetilde{U_{\overline{k}}}\widehat{\widetilde{T}}=T W^{-1}\widetilde{U_{\overline{k}}}W \widehat{T^{-1}}=T \widetilde{U_{\overline{k}}}\widehat{T^{-1}}=\widehat{T}\widetilde{U_{\overline{k}}} T^{-1},\qquad k=0,2,
\\
(\widehat{T} U_{\overline{k}} T^{-1}){\widehat{\;\;}}&=&T\widehat{U_{\overline{k}}}\widehat{T^{-1}}=\widehat{T}\widehat{U_{\overline{k}}}T^{-1},\qquad k=0,1,2,3,
\end{eqnarray*}
where we use that $W$ commutes with all even elements. Using the definition of quaternion types (\ref{qtdef}), we obtain $\widehat{T}U_{\overline{k}}T^{-1}\in\C^{\overline{k}}$ for $k=0,2$. Thus, $\Q'\subseteq\check{\Gamma}^{\overline{k}}$, $k=0,2$.

Let us prove $\check{\Gamma}^{\overline{1}}\subseteq\Q^{\pm}$. Suppose $T\in\check{\Gamma}^{\overline{1}}$; then $\widehat{T}e_a T^{-1}\in\C^{\overline{1}}$, $a=1,\ldots,n$. We get
\begin{eqnarray}\label{TeaTov1}
\widehat{T} e_a T^{-1}=(\widehat{T} e_a T^{-1}){\widetilde{\;\;}}=\widetilde{T^{-1}}e_a\widehat{\widetilde{T}}
\end{eqnarray}
using the property of the reversion (\ref{^uv=^u^v}). Multiplying both sides of the equation (\ref{TeaTov1}) on the left by $\widetilde{T}$, on the right by $\widehat{\widetilde{T^{-1}}}$, we obtain
\begin{eqnarray*}
\widehat{(\widehat{\widetilde{T}}T)} e_a (\widehat{\widetilde{T}}T)^{-1}=e_a\qquad a=1,\ldots,n;
\end{eqnarray*}
therefore, $\check{\ad}_{\widehat{\widetilde{T}}T}(e_a)=e_a$, $a=1,\ldots,n.$ Thus, $\widehat{\widetilde{T}}T\in\ker(\check{\ad})=\C^{0\times}$ (\ref{kerchad}) and $\check{\Gamma}^{\overline{1}}\subseteq\B_{\pm}$. Also we get
\begin{eqnarray}\label{TeaTov12}
\widehat{T} e_a T^{-1}=-(\widehat{T} e_a T^{-1}){\widehat{\widetilde{\;\;}}}=\widehat{\widetilde{T^{-1}}} e_a\widetilde{T}
\end{eqnarray}
using the property of the Clifford conjugation (\ref{^uv=^u^v}). Multiplying both sides of the equation (\ref{TeaTov12}) on the left by $\widehat{\widetilde{T}}$, on the right by $\widetilde{T^{-1}}$, we obtain
\begin{eqnarray*}
\widehat{(\widetilde{T}T)}e_a(\widetilde{T}T)^{-1}=e_a\qquad a=1,\ldots, n;
\end{eqnarray*}
hence, $\check{\ad}_{\widetilde{T}T}(e_a)=e_a$, $a=1,\ldots,n.$ Thus, $\widetilde{T}T\in\ker(\check{\ad})=\C^{0\times}$ (\ref{kerchad}) and $\check{\Gamma}^{\overline{1}}\subseteq\A_{\pm}$. Using Lemma \ref{lemmaQABP}, we obtain $\check{\Gamma}^{\overline{1}}\subseteq\A_{\pm}\cap\B_{\pm}=\Q^{\pm}$.

In the exceptional cases  $n=1,2$, we have $\check{\Gamma}^{\overline{3}}=\C^{\times}\not\subset\Q^{\pm}$.

Let us prove $\check{\Gamma}^{\overline{3}}\subseteq\Q^{\pm}$ if $n\geq3$. In the case $n=3$, we have $\check{\Gamma}^{\overline{3}}=\check{\Gamma}^3=\P^{\pm}=\Q^{\pm}$ by Lemmas \ref{lemmag0'gn'} and \ref{lemmaQABP}. Consider the cases $n\geq4$. Suppose $\widehat{T}\C^{\overline{3}}T^{-1}\subseteq\C^{\overline{3}}$; then
\begin{eqnarray}\label{tut3}
\widehat{T} U T^{-1}=-(\widehat{T} U T^{-1}){\widetilde{\;\;}}=\widetilde{T^{-1}}U\widehat{\widetilde{T}}\qquad \forall U\in\C^{\overline{3}}.
\end{eqnarray}
Multiplying both sides of the equation (\ref{tut3}) on the left by $\widetilde{T}$, on the right by $T$, we obtain
\begin{eqnarray}\label{ad^_TT3}
\widehat{(\widehat{\widetilde{T}}T)}U=U(\widehat{\widetilde{T}}T)\qquad \forall U\in\C^{\overline{3}}.
\end{eqnarray}
Since  $n\geq4$, any generator $e_a$, $a=1,\ldots,n$, can always be represented as the product of three elements of grade $3$. For example, we can get the generator $e_1$ by the multiplication of $e_{123}$, $e_{124}$, $e_{134}$. Using Lemma  \ref{modd^XU=UX}, we conclude that all generators satisfy  (\ref{ad^_TT3}) and $\check{\ad}_{\widehat{\widetilde{T}}T}(e_a)=e_a$, $a=1,\ldots,n.$ Therefore, 
$\widehat{\widetilde{T}}T\in\ker(\check{\ad})=\C^{0\times}$ (\ref{kerchad}) and $\check{\Gamma}^{\overline{3}}\subseteq\B_{\pm}$. Also we have
\begin{eqnarray*}
\widehat{T} U T^{-1}=(\widehat{T} U T^{-1}){\widehat{\widetilde{\;\;}}}=\widehat{\widetilde{T^{-1}}} U\widetilde{T}\qquad \forall U\in\C^{\overline{3}},
\end{eqnarray*}
and $\check{\ad}_{\widetilde{T}T}(U)=U$ for any $U\in\C^{\overline{3}}.$ Thus, $\widetilde{T}T\in\ker(\check{\ad})=\C^{0\times}$ (\ref{kerchad}) and $\check{\Gamma}^{\overline{3}}\subseteq\A_{\pm}$. Using Lemma \ref{lemmaQABP}, we obtain $\check{\Gamma}^{\overline{3}}\subseteq\A_{\pm}\cap\B_{\pm}=\Q^{\pm}$.

In the exceptional case $n=1$, we have $\check{\Gamma}^{\overline{2}}=\C^{\times}\not\subset\Q^{\pm}$. In the case $n=2$, we have $\check{\Gamma}^{\overline{2}}=\check{\Gamma}^2=\C^{\times}\not\subset\Q^{\pm}$ by Lemma \ref{lemmag0'gn'}. 

Let us prove $\check{\Gamma}^{\overline{2}}\subseteq\Q^{\pm}$ if $n=1,2,3\mod{4}$, $n\geq3$, and $\check{\Gamma}^{\overline{2}}\subseteq\Q'$ if $n=0\mod{4}$, $n\geq4$. Suppose $\widehat{T}\C^{\overline{2}}T^{-1}\subseteq\C^{\overline{2}}$; then we obtain
\begin{eqnarray}\label{tut2}
\widehat{T} U T^{-1}=-(\widehat{T} U T^{-1}){\widetilde{\;\;}}=\widetilde{T^{-1}}U\widehat{\widetilde{T}}\qquad \forall U\in\C^{\overline{2}}
\end{eqnarray}
using the property of the reversion (\ref{^uv=^u^v}). Multiplying both sides of the equation (\ref{tut2}) on the left by $\widetilde{T}$, on the right by $T$, we get
\begin{eqnarray*}
\widehat{(\widehat{\widetilde{T}}T)} U =U(\widehat{\widetilde{T}}T)\qquad \forall U\in\C^{\overline{2}}.
\end{eqnarray*}
Since  $n\geq3$, any even basis element can be represented as the product of an odd number of grade-2 elements. For example, we can get the element $e$ by the multiplication of $e_{12}$, $e_{13}$, $e_{23}$; the element $e_{1234}$ by the multiplication of $e_{12}$, $e_{13}$, $e_{14}$. Using Lemma \ref{modd^XU=UX}, we obtain
\begin{eqnarray*}
\widehat{(\widehat{\widetilde{T}}T)} U =U(\widehat{\widetilde{T}}T)\qquad \forall U\in\C^{(0)}.
\end{eqnarray*}
We get $\widehat{\widetilde{T}}T\in\C^{0}\oplus\C^{n}$ if $n$ is even, $\widehat{\widetilde{T}}T\in\C^{0}$ if $n$ is odd using Lemma \ref{an1.3lemma}. Also we have
\begin{eqnarray*}
\widehat{T} U T^{-1}=-(\widehat{T} U T^{-1}){\widehat{\widetilde{\;\;}}}=\widehat{\widetilde{T^{-1}}}U \widetilde{T}\qquad \forall U\in\C^{\overline{2}};
\end{eqnarray*}
therefore, $\widehat{(\widetilde{T}T)}U=U(\widetilde{T}T)$ for any $U\in\C^{\overline{2}}$; hence, $\widehat{(\widetilde{T}T)}U=U(\widetilde{T}T)$ for any $U\in\C^{(0)}.$ Using Lemma \ref{an1.3lemma}, we get $\widetilde{T}T\in(\C^{0}\oplus\C^{n})^{\times}$ if $n$ is even, $\widetilde{T}T\in\C^{0\times}$ if $n$ is odd. Since $\widetilde{T}T\in\C^{\overline{01}\times}$,  $\widehat{\widetilde{T}}T\in\C^{\overline{03}\times}$ (\ref{nfunc}), in the cases $n=1,2,3\mod{4}$, we have $\widetilde{T}T\in\C^{0\times}$, $\widehat{\widetilde{T}}T\in\C^{0\times}$, so $T\in\A_{\pm}\cap\B_{\pm}=\Q^{\pm}$ by Lemma \ref{lemmaQABP}.
In the case $n=0\mod{4}$, we have $T\in\A'\cap\B'=\Q'$ (\ref{a'b'q'}).

Let us prove $\check{\Gamma}^{\overline{0}}\subseteq\Q^{\pm}$  in the cases  $n=1,2,3\mod{4}$ and $\check{\Gamma}^{\overline{0}}\subseteq\Q'$ in the case $n=0\mod{4}$. Suppose $\widehat{T}\C^{\overline{0}}T^{-1}\subseteq\C^{\overline{0}}$; then we obtain
\begin{eqnarray*}
\widehat{T}U T^{-1}=(\widehat{T}U T^{-1}){\widetilde{\;\;}}=\widetilde{T^{-1}}U\widehat{\widetilde{T}},\qquad  \widehat{T} U T^{-1}=(\widehat{T} U T^{-1}){\widehat{\widetilde{\;\;}}}=\widehat{\widetilde{T^{-1}}}U\widetilde{T}\qquad \forall U\in\C^{\overline{0}}
\end{eqnarray*}
using the properties of the reversion and the Clifford conjugation (\ref{^uv=^u^v}). Hence,
\begin{eqnarray}\label{u_000}
\widehat{(\widehat{\widetilde{T}}T)} U =U(\widehat{\widetilde{T}}T),\qquad \widehat{(\widetilde{T}T)}U=U(\widetilde{T}T)\qquad\forall U\in\C^{\overline{0}}.
\end{eqnarray}
Consider the cases $n\geq5$. Any even basis element can be represented as the product of an odd number of elements of the subspace $\C^{\overline{0}}$. For instance, we can get the element $e_{12}$ by the multiplication of $e_{1235}$, $e_{1245}$, $e_{1234}$. Therefore, using Lemma \ref{modd^XU=UX}, we obtain
\begin{eqnarray*}
\widehat{(\widehat{\widetilde{T}}T)} U =U(\widehat{\widetilde{T}}T),\qquad \widehat{(\widetilde{T}T)}U=U(\widetilde{T}T)\qquad \forall U\in\C^{(0)}.
\end{eqnarray*}
Using Lemma \ref{an1.3lemma}, we get
\begin{eqnarray}\label{totn}
\widehat{\widetilde{T}}T,\widetilde{T}T\in
\left\lbrace
\begin{array}{lll}
(\C^0\oplus\C^n)^{\times}&\mbox{if $n$ is even},&\label{2.31_}
\\
\C^{0\times}&\mbox{if $n$ is odd}.&\label{3.31_}
\end{array}
\right.
\end{eqnarray}
Since $\widehat{\widetilde{T}}T\in\C^{\overline{03}\times}$, $\widetilde{T}T\
\in\C^{\overline{01}\times}$ (\ref{nfunc}), in the cases $n=1,2,3\mod{4}$,  it follows from (\ref{totn}) that $\widehat{\widetilde{T}}T\in\C^{0\times}$, $\widetilde{T}T\in\C^{0\times}$; therefore, $T\in\A_{\pm}\cap\B_{\pm}=\Q^{\pm}$ (Lemma \ref{lemmaQABP}).
In the case $n=0\mod{4}$, it follows from (\ref{totn}) that $T\in\A'\cap\B'=\Q'$ (\ref{a'b'q'}).
Consider the case $n=4$. 
Substituting $U=e\in\C^{\overline{0}}$ into the equations (\ref{u_000}), we get $\widehat{\widetilde{T}}T\in\C^{(0)\times}=\C^{024\times}$, $\widetilde{T}T\in\C^{024\times}$. Since $\widehat{\widetilde{T}}T\in\C^{\overline{03}\times}$, $\widetilde{T}T\in\C^{\overline{01}\times}$ (\ref{nfunc}), we obtain (\ref{2.31_}) again and $T\in\A'\cap\B'=\Q'$.
In the cases $n\leq 3$, we have $\check{\Gamma}^{\overline{0}}=\check{\Gamma}^{0}=\P^{\pm}=\Q^{\pm}$ by Lemmas \ref{lemmag0'gn'} and \ref{lemmaQABP}. 
\end{proof}

\begin{remark}\label{q*qsmall}
Using Theorem \ref{maintheoremQ*} from this paper and Theorem 6.3 from \cite{OnInner}, we get
\begin{eqnarray}\label{gkpgkp}
\Gamma^{\overline{k}}\subseteq\P\qquad n\geq4,\qquad\check{\Gamma}^{\overline{k}}\subseteq\P^{\pm}\qquad n\geq3,\qquad k=0,1,2,3,
\end{eqnarray}
since  $\Q\subseteq\P$, $\Q^{\pm}\subseteq\P^{\pm}$ in the case of arbitrary $n\geq1$ (Lemma 6.1 from \cite{OnInner} and Lemma \ref{lemmaQABP}),  $\Q'\subseteq\P^{\pm}$ in the case $n=0\mod{4}$ (Lemma \ref{lemmaq*q'pp*}).
Using (\ref{gkpgkp}), we give the equivalent definitions of the groups $\Gamma^{\overline{k}}$ and $\check{\Gamma}^{\overline{k}}$, $k=0,1,2,3$,  respectively:
\begin{eqnarray}
&&\Gamma^{\overline{k}}:=\{T\in\Z^{\times}(\C^{(0)\times}\cup\C^{(1)\times}):\quad T\C^{\overline{k}}T^{-1}\subseteq\C^{\overline{k}}\},\qquad k=0,1,2,3,\qquad n\neq2,3;\label{gk1}
\\
&&\check{\Gamma}^{\overline{k}}:=\{T\in\C^{(0)\times}\cup\C^{(1)\times}:\quad T\C^{\overline{k}}T^{-1}\subseteq\C^{\overline{k}}\},\qquad k=0,1,2,3,\qquad n\geq3.\label{gk2}
\end{eqnarray}
Note that if $n=2,3$, then the definition (\ref{gk1}) is true for $k=1,2$; if $n=1,2$, then the definition  (\ref{gk2}) is true for $k=0,1$.
Using the definitions (\ref{gk1}) and (\ref{gk2}), we conclude that the groups $\Gamma^{\overline{k}}$ and $\check{\Gamma}^{\overline{k}}$ are related in the following way
\begin{eqnarray}\label{rel_gk_ov}
\Gamma^{\overline{k}}=\Z^{\times}\check{\Gamma}^{\overline{k}},\qquad k=0,1,2,3,\qquad n\geq4.
\end{eqnarray}
It follows from (\ref{rel_gk_ov}) that in the cases of $n\geq4$
\begin{eqnarray}
&&\Q^{\pm}=\check{\Gamma}^{\overline{0}}=\check{\Gamma}^{\overline{1}}=\check{\Gamma}^{\overline{2}}=\check{\Gamma}^{\overline{3}}\subset\Gamma^{\overline{0}}=\Gamma^{\overline{1}}=\Gamma^{\overline{2}}=\Gamma^{\overline{3}}=\Q=\Q'\subset\C^{\times},\qquad n=1,3\mod{4};
\\
&&\Q^{\pm}=\check{\Gamma}^{\overline{0}}=\check{\Gamma}^{\overline{1}}=\check{\Gamma}^{\overline{2}}=\check{\Gamma}^{\overline{3}}=\Gamma^{\overline{0}}=\Gamma^{\overline{1}}=\Gamma^{\overline{2}}=\Gamma^{\overline{3}}=\Q=\Q'\subset\C^{\times},\qquad n=2\mod{4};
\\
&&\Q^{\pm}=\check{\Gamma}^{\overline{1}}=\check{\Gamma}^{\overline{3}}=\Gamma^{\overline{1}}=\Gamma^{\overline{3}}=\Q\subseteq\Q'=\check{\Gamma}^{\overline{0}}=\check{\Gamma}^{\overline{2}}=\Gamma^{\overline{0}}=\Gamma^{\overline{2}}\subset\C^{\times},\qquad n=0\mod{4}.\label{n04q}
\end{eqnarray}
In the exceptional cases, we have
\begin{eqnarray}
&&\Q^{\pm}=\check{\Gamma}^{\overline{0}}=\check{\Gamma}^{\overline{1}}\subset\Gamma^{\overline{0}}=\Gamma^{\overline{1}}=\Gamma^{\overline{2}}=\Gamma^{\overline{3}}=\check{\Gamma}^{\overline{2}}=\check{\Gamma}^{\overline{3}}=\Q=\Q'=\C^{\times},\qquad n=1;\label{n1q}
\\
&&\Q^{\pm}=\check{\Gamma}^{\overline{0}}=\check{\Gamma}^{\overline{1}}=\Gamma^{\overline{1}}=\Gamma^{\overline{2}}=\Q=\Q'\subset \check{\Gamma}^{\overline{2}}=\Gamma^{\overline{0}}=\Gamma^{\overline{3}}=\check{\Gamma}^{\overline{3}}=\C^{\times},\qquad n=2;\label{n2q}
\\
&&\Q^{\pm}=\check{\Gamma}^{\overline{0}}=\check{\Gamma}^{\overline{1}}=\check{\Gamma}^{\overline{2}}=\check{\Gamma}^{\overline{3}}\subset\Gamma^{\overline{1}}=\Gamma^{\overline{2}}=\Q=\Q'\subset \Gamma^{\overline{0}}=\Gamma^{\overline{3}}=\C^{\times},\qquad n=3.\label{n3q}
\end{eqnarray}
\end{remark}

\section{The groups $\check{\Gamma}^{\lowercase{k}}$}\label{sectiongk'}

Let us use the following notation for the groups preserving the subspaces of fixed grades or their direct sums under the adjoint representation:
\begin{eqnarray*}
&&\Gamma^{k}:=\{T\in\C^{\times}:\quad T\C^{k}T^{-1}\subseteq\C^{k} \},\quad k=0,1,\ldots,n;
\\
&&\Gamma^{kl}:=\{T\in\C^{\times}:\quad T\C^{kl}T^{-1}\subseteq\C^{kl} \},\quad k,l=0,1,\ldots,n;
\end{eqnarray*}
and the twisted adjoint representation:
\begin{eqnarray}
&&\check{\Gamma}^{k}:=\{T\in\C^{\times}:\quad \widehat{T}\C^{k}T^{-1}\subseteq\C^{k} \},\quad k=0,1,\ldots,n;\label{x10}
\\
&&\check{\Gamma}^{kl}:=\{T\in\C^{\times}:\quad \widehat{T}\C^{kl}T^{-1}\subseteq\C^{kl} \},\quad k,l=0,1,\ldots,n.
\end{eqnarray}
Let us note that in Section \ref{sectionP*}, we have already considered the groups $\Gamma^{0}$, $\Gamma^{n}$, $\Gamma^{0n}$, $\check{\Gamma}^{0}$, $\check{\Gamma}^{n}$, $\check{\Gamma}^{0n}$ preserving the subspaces of grades $0$, $n$ and their direct sum under $\ad$ and $\check{\ad}$. We have also shown their relation with $\P^{\pm}$ and $\P$ (Lemmas \ref{lemmag0'gn'}, \ref{Gamma0n}).

\begin{lemma}\label{gkq*gkq*gkq'}
We have
\begin{eqnarray}
\check{\Gamma}^{k}&\subseteq&\Q^{\pm},\qquad k=1,2,3,\ldots,n-1,\qquad n=1,2,3\mod{4},\label{f1gq}
\\
\check{\Gamma}^{k}&\subseteq&\Q^{\pm},\qquad k=1,3,5,\ldots,n-1,\qquad n=0\mod{4},\label{f2gq}
\\
\check{\Gamma}^{k}&\subseteq&\Q',\qquad k=2,4,6,\ldots,n-2,\qquad n=0\mod{4}.
\end{eqnarray}
\end{lemma}
\begin{proof}
Suppose $\widehat{T} U_k T^{-1}\in\C^k$ for any $U_k\in\C^k$ with some fixed $k$. 
We obtain
\begin{eqnarray}\label{twok}
\widehat{T}U_k T^{-1}=\pm(\widehat{T}U_k T^{-1}){\widetilde{\;\;}}=\widetilde{T^{-1}}U_k\widehat{\widetilde{T}},\qquad \widehat{T}U_k T^{-1}=\pm(\widehat{T}U_k T^{-1}){\widehat{\widetilde{\;\;}}}=\widehat{\widetilde{T^{-1}}}U_k\widetilde{T},\qquad \forall U_k\in\C^k
\end{eqnarray}
using the properties of the reversion and the Clifford conjugation (\ref{^uv=^u^v}). We multiply both sides of the first equation in (\ref{twok}) on the left by $\widetilde{T}$ and on the right by $T$, both sides of the second equation in (\ref{twok}) on the left by $\widehat{\widetilde{T}}$ and on the right by $T$, and get
\begin{eqnarray}\label{TuutTuuT}
\widehat{(\widehat{\widetilde{T}}T)}U_k=U_k(\widehat{\widetilde{T}}T),\qquad \widehat{(\widetilde{T}T)}U_k=U_k(\widetilde{T}T),\qquad \forall U_k\in\C^k.
\end{eqnarray}

If $k$ is odd and $k\neq n$, then any generator $e_a$, $a=1,\ldots,n$, can always be represented as the product of $k$ elements of grade $k$. For instance, we can get the generator $e_1$ by the multiplication of $e_{123}$, $e_{124}$, $e_{134}$ ($k=3$); $e_{12345}$, $e_{12346}$, $e_{12356}$, $e_{12456}$, $e_{13456}$ ($k=5$). Then, using Lemma \ref{modd^XU=UX}, we obtain
\begin{eqnarray*}
\widehat{(\widehat{\widetilde{T}}T)}e_a=e_a(\widehat{\widetilde{T}}T),\qquad \widehat{(\widetilde{T}T)}e_a=e_a(\widetilde{T}T),\qquad a=1,\ldots,n;
\end{eqnarray*}
therefore, $\widehat{\widetilde{T}}T$, $\widetilde{T}T\in\ker(\check{\ad})=\C^{0\times}$ (\ref{kerchad}). Thus, $\check{\Gamma}^{k}\subseteq\A_{\pm}$, $\check{\Gamma}^{k}\subseteq\B_{\pm}$. Using Lemma \ref{lemmaQABP}, we get $\check{\Gamma}^{k}\subseteq\A_{\pm}\cap\B_{\pm}=\Q^{\pm}$.

If $k$ is even and $k\neq n$, $k\neq0$, then we can always represent the identity element $e$ as the product of an odd number of grade-$k$ elements. For example, 
we can get $e$ by the multiplication of $e_{12}$, $e_{13}$, $ e_{23}$ ($k=2$); $ e_{1234}$, $e_{1235}$, $e_{1245}$, $e_{1345}$, $e_{2345}$ ($k=4$). Using Lemma \ref{modd^XU=UX}, we conclude that $U_k=e$ satisfies the equations in (\ref{TuutTuuT}). Hence, $\widetilde{T}T\in\C^{(0)\times}$, $\widehat{\widetilde{T}}T\in\C^{(0)\times}$, and
\begin{eqnarray*}
(\widehat{\widetilde{T}}T)U_k=U_k(\widehat{\widetilde{T}}T),\qquad (\widetilde{T}T)U_k=U_k(\widetilde{T}T),\qquad \forall U_k\in\C^k.
\end{eqnarray*}
We can always represent any basis element of grade $2$ as the product of grade-$k$ elements; therefore, 
\begin{eqnarray*}
(\widehat{\widetilde{T}}T)e_{ab}=e_{ab}(\widehat{\widetilde{T}}T),\qquad (\widetilde{T}T)e_{ab}=e_{ab}(\widetilde{T}T),\qquad \forall a<b.
\end{eqnarray*}
Using Lemma \ref{an1.3lemma}, we obtain $\widehat{\widetilde{T}}T,\widetilde{T}T\in(\C^{0}\oplus\C^{n})^{\times}$. Since $\widehat{\widetilde{T}}T,\widetilde{T}T \in\C^{(0)\times}$, $\widehat{\widetilde{T}}T\in\C^{\overline{03}\times}$, and $\widetilde{T}T\in\C^{\overline{01}\times}$ (\ref{nfunc}), we have $\widehat{\widetilde{T}}T,\widetilde{T}T\in\C^{0\times}$ if $ n=1,2,3\mod{4}$, $\widehat{\widetilde{T}}T,\widetilde{T}T\in(\C^{0}\oplus\C^{n})^{\times}$  if $n=0\mod{4}$.
Thus we have $\check{\Gamma}^{k}\subseteq\A_{\pm}\cap\B_{\pm}=\Q^{\pm}$  in the cases $n=1,2,3\mod{4}$ (Lemma \ref{lemmaQABP}),  $\check{\Gamma}^{k}\subseteq\A'\cap\B'=\Q'$ in the case $n=0\mod{4}$ (\ref{a'b'q'}). 
\end{proof}

\begin{remark}
Note that
\begin{eqnarray*}
\check{\Gamma}^k\not\subset\Q^{\pm},\qquad k=2,4,6,\ldots,n-2,\qquad n=0\mod{4}.
\end{eqnarray*}
Let us give the following example. Consider the element $T=e + 2 e_{1234}$ in the case $n=4$. It is invertible because $(e + 2e_{1234})(e - 2 e_{1234})=e - 4(e_{1234})^2\in\C^{0\times}.$ We have $T\in\check{\Gamma}^2$, since
\begin{eqnarray*}
\widehat{T}e_{ab} T^{-1}=T e_{ab} T^{-1}=e_{ab}TT^{-1}=e_{ab}\in\C^2,\qquad\forall  a<b,
\end{eqnarray*}
where we use that $e_{1234}$ commutes with all even elements. We have
\begin{eqnarray*}
\widetilde{T}T=(e + 2 e_{1234})(e + 2 e_{1234})=e + 4 e_{1234} + 4(e_{1234})^2\not\in\C^{0\times};
\end{eqnarray*}
thus, $T\not\in\Q^{\pm}$.
\end{remark}

\begin{lemma}\label{lemmaGammaP}
We have
\begin{eqnarray*}
\check{\Gamma}^{k}&\subseteq&\P^{\pm},\qquad k=1,\ldots,n-1.
\end{eqnarray*}
\end{lemma}
\begin{proof}
We have $\Q'\subseteq\P^{\pm}$ in the case $n=0\mod{4}$ by Lemma \ref{lemmaq*q'pp*}, $\Q^{\pm}\subseteq\P^{\pm}$ in the case of arbitrary $n\geq1$ by Lemma \ref{lemmaQABP}. Using Lemma \ref{gkq*gkq*gkq'}, we obtain $\check{\Gamma}^{k}\subseteq\P^{\pm}$ if $k\neq0$, $k\neq n$.
\end{proof}

Using the previous lemma for $k=1,\ldots,n-1$ and Lemma \ref{lemmag0'gn'} for $k=0$, we give the equivalent definition of the groups $\check{\Gamma}^{k}$:
\begin{eqnarray}\label{eqgk'}
\check{\Gamma}^{k} &:=& \{T\in\C^{(0)\times}\cup\C^{(1)\times}:\quad T\C^{k} T^{-1}\subseteq\C^{k}\},\qquad k=0,1,\ldots,n-1.
\end{eqnarray}
Using the definition (\ref{eqgk'}), we conclude that the groups $\Gamma^k$ and $\check{\Gamma}^k$, $k=1,\ldots,n-1$, are related in the following way:
\begin{eqnarray}\label{rel_gk}
\Gamma^{k}=\Z^{\times}\check{\Gamma}^{k},\qquad k=1,\ldots,n-1.
\end{eqnarray}
Note that
\begin{eqnarray*}
\check{\Gamma}^{k}\subseteq\Gamma^{k},\qquad k=0,1,\ldots,n-1,
\end{eqnarray*}
in particular, these groups coincide if $n$ is even:
\begin{eqnarray}\label{gkgkeven}
\check{\Gamma}^{k}=\Gamma^{k},\qquad k=1,\ldots,n-1,\qquad \mbox{$n$ is even}.
\end{eqnarray}
Note also that in the case $k=1$, from the definition (\ref{eqgk'}) it follows that the Lipschitz group definitions (\ref{Lip1}) and (\ref{Lip2}) are equivalent.
\\

We know that the Clifford group is a subgroup of the groups $\Gamma^k$, $k=0,1,\ldots,n$, preserving the subspaces of fixed grades under the adjoint representation (see, for example, Theorem 2.2 from \cite{OnInner}):
\begin{eqnarray}\label{ggk}
\Gamma\subseteq\Gamma^k,\quad k=0,1,\ldots,n.
\end{eqnarray}
In the next lemma, we prove the similar statement about the Lipschitz group and the groups $\check{\Gamma}^{k}$.

\begin{lemma}\label{g+-gk'}
We have
\begin{eqnarray*}
\check{\Gamma}^{1}\subseteq\check{\Gamma}^{k},\qquad k=0,1,\ldots,n.
\end{eqnarray*}
As a consequence, 
\begin{eqnarray*} 
\check{\Gamma}^{1}=\bigcap_{k=0}^n\check{\Gamma}^{k}=\{T\in\C^{(0)\times}\cup\C^{(1)\times}:\quad T\C^k T^{-1}\subseteq\C^k,\quad k=0,1,\ldots,n\}.
\end{eqnarray*}
\end{lemma}
\begin{proof}
We have $\check{\Gamma}^1\subseteq\Gamma\subseteq\Gamma^k$, $k=0,1,\ldots n$, by (\ref{ggk}) and the definitions of the groups $\check{\Gamma}^1$ (\ref{Lip2}) and $\Gamma$ (\ref{Cl2}). 
Since $\check{\Gamma}^1\subseteq\P^{\pm}$ by Lemma \ref{lemmaGammaP}, we get $\check{\Gamma}^{1}\subseteq\Gamma^k\cap\P^{\pm}=\check{\Gamma}^k$, $k=0,1,\ldots n$.
\end{proof}

The Clifford group coincides with the group $\Q$ in the cases $n\leq5$ (this statement is proved in Theorem 7.3 from \cite{OnInner}):
\begin{eqnarray}\label{gq56OnInner}
\Gamma=\Q,\quad n\leq5;\qquad \Gamma\neq\Q,\quad n=6.
\end{eqnarray} 
In the next lemma, we prove that the Lipschitz group and the group $\Q^{\pm}$ are related in a similar way.

\begin{lemma}\label{glq*}
We have
\begin{eqnarray}\label{falp}
\check{\Gamma}^{1}=\Q^{\pm},\quad n\leq5;\qquad \check{\Gamma}^{1}\neq\Q^{\pm},\quad n=6.
\end{eqnarray}
\end{lemma}
\begin{proof}
We have $\check{\Gamma}^{1}\subseteq\Q^{\pm}$ in the case of arbitrary $n\geq1$ by Lemma \ref{gkq*gkq*gkq'}. Consider the cases $n\leq5$. Let us prove $\Q^{\pm}\subseteq\check{\Gamma}^{1}$. Since $\Q^{\pm}\subseteq\Q$ (\ref{Q}), we obtain $\Q^{\pm}\subseteq\Gamma$ using (\ref{gq56OnInner}). Since $\Q^{\pm}\subseteq\P^{\pm}$ by Lemma \ref{lemmaQABP}, we obtain $\Q^{\pm}\subseteq\Gamma\cap\P^{\pm}=\check{\Gamma}^{1}$. 

In the case $n=p=6$, $q=0$, consider the element
\begin{eqnarray*}
T=\frac{1}{\sqrt{2}}(e_{12} + e_{3456})\in\C^{(0)\times}.
\end{eqnarray*}
We have
\begin{eqnarray*}
&&\widetilde{T}T=\frac{1}{2}(-e_{12} + e_{3456})(e_{12}+e_{3456})=e,
\\
&&\widehat{T}e_1 T^{-1}=\frac{1}{2}(e_{12}+e_{3456})e_1(-e_{12}+e_{3456})=-e_{23456}\not\in\C^{1},
\end{eqnarray*}
i.e. $T\in\Q^{\pm}$, $T\not\in\check{\Gamma}^{1}$; thus, $\check{\Gamma}^{1}\neq\Q^{\pm}$. 
\end{proof}

\begin{lemma}\label{gkgnk}
We have
\begin{eqnarray*}
\check{\Gamma}^{k}=\check{\Gamma}^{n-k},\qquad k=1,\ldots,n-1.
\end{eqnarray*}
\end{lemma}
\begin{proof}
Let us prove that the condition $T\C^{k}T^{-1}\subseteq\C^k$ is equivalent to $T\C^{n-k}T^{-1}\subseteq\C^{n-k}$, $k=1,\ldots,n-1$. We can multiply both sides of these equations by $e_{1\ldots n}$. In the case of odd $n$, we use $e_{1\ldots n}\in\Z^{\times}$; in the case of even $n$, we use that $T\in\C^{(0)\times}\cup\C^{(1)\times}$ (\ref{eqgk'}) and $e_{1\ldots n}$ commutes with all even elements and anticommutes with all odd elements. Also we apply the fact $\C^{k}e_{1\ldots n}=\C^{n-k}$. 
\end{proof}

Let us remark that we have the similar statement about the groups $\Gamma^{k}$ from \cite{OnInner}
\begin{eqnarray*}
\Gamma^k=\Gamma^{n-k},\qquad k=1,\ldots,n-1.
\end{eqnarray*}

\section{The cases of small dimensions}
\label{sectSmall}

\begin{lemma}
In the cases $n\leq4$, we have the following (two, two, three, and three respectively) different groups :  
\begin{eqnarray*}
&&\check{\Gamma}^{0}=\check{\Gamma}^1=\P^{\pm}=\A_{\pm}=\Q^{\pm}\subset\Gamma^{0}=\Gamma^{1}=\Q=\Q'=\C^{\times},\qquad n=1;\label{small1}
\\
&&\check{\Gamma}^{0}=\check{\Gamma}^1=\Gamma^{1}=\Gamma^{2}=\P^{\pm}=\A_{\pm}=\Q^{\pm}=\Q=\Q'\subset \check{\Gamma}^{2}=\Gamma^{0}=\C^{\times},\qquad n=2;\label{small2}
\\
&&\check{\Gamma}^{0}=\check{\Gamma}^1=\check{\Gamma}^2=\check{\Gamma}^3=\Q^{\pm}=\P^{\pm}\subset\Gamma^{1}=\Gamma^{2}=\Q=\Q'\subset\Gamma^{0}=\Gamma^{3}=\C^{\times},\qquad n=3;\label{small3}
\\
&&\check{\Gamma}^1=\check{\Gamma}^3=\Gamma^{1}=\Gamma^{3}=\Q^{\pm}=\Q\subset\check{\Gamma}^{0}=\check{\Gamma}^{2}=\Gamma^2=\Gamma^4=\Q'=\P^{\pm}\subset\Gamma^{0}=\check{\Gamma}^4=\C^{\times},\qquad n=4.\label{small4}
\end{eqnarray*}
\end{lemma}
\begin{proof}
We have $\check{\Gamma}^0=\P^{\pm}$ in the case of arbitrary $n\geq1$ and $\check{\Gamma}^{n}=\C^{\times}$ in the case of even $n$ by Lemma \ref{lemmag0'gn'}, $\Gamma^{0}=\C^{\times}$ by (\ref{g0gnoninner}). We have $\Q^{\pm}=\P^{\pm}$ if $n\leq3$ and $\Q^{\pm}=\A_{\pm}$ if $n=1,2$ by Lemma \ref{lemmaQABP}. In the case $n=4$,  $\Q'=\P^{\pm}$ by Lemma \ref{lemmaq*q'pp*}. The other statements follow from  (\ref{n04q}) - (\ref{n3q}) (see Remark \ref{q*qsmall}).
\end{proof}

\begin{lemma}
In the case $n=5$, we have the following four different groups
\begin{eqnarray}
\check{\Gamma}^1=\check{\Gamma}^2=\check{\Gamma}^3=\check{\Gamma}^4=\Q^{\pm},&\quad&\Gamma^1=\Gamma^2=\Gamma^3=\Gamma^4=\Q=\Q',\label{n5_1}\\
\check{\Gamma}^{0}=\check{\Gamma}^{5}=\P^{\pm},&\quad&
\Gamma^{0}=\Gamma^{5}=\C^{\times},\label{n5_2}
\end{eqnarray}
and
\begin{eqnarray}
\Q^{\pm}\subset\P^{\pm}\subset\C^{\times},\qquad \Q^{\pm}\subset\Q=\Q'\subset\C^{\times}.
\end{eqnarray}
\end{lemma}
\begin{proof}
We have $\check{\Gamma}^1=\Q^{\pm}$ if $n=5$ by Lemma \ref{glq*}.
Since $\check{\Gamma}^1\subseteq\check{\Gamma}^{k}\subseteq\Q^{\pm}$ for some fixed $k=1,2,3,4$ (Lemmas \ref{gkq*gkq*gkq'} and \ref{g+-gk'}), we obtain $\check{\Gamma}^1=\check{\Gamma}^{k}=\Q^{\pm}$, and the proof of the first four equalities in (\ref{n5_1}) is completed.
We have the first two equalities in (\ref{n5_2}) by Lemma \ref{lemmag0'gn'}. The last five equalities in (\ref{n5_1}) and the last two equalities in (\ref{n5_2}) are proved in the paper \cite{OnInner}. 
The four considered groups are different by (\ref{QABPcap4}), (\ref{rel_Q}), (\ref{Q}), and (\ref{P}).
\end{proof}

\begin{lemma}\label{n=6gg'}
In the case $n=6$, we have the following four different groups 
\begin{eqnarray}
&\Gamma^{0}=\check{\Gamma}^6=\C^{\times},\qquad\check{\Gamma}^{0}=\Gamma^{6}=\P^{\pm},\qquad\Gamma^{1}=\Gamma^{5}=\check{\Gamma}^1=\check{\Gamma}^5=\Gamma^{\pm},\label{n6_1}
\\
&\Gamma^{2}=\Gamma^{3}=\Gamma^{4}=\check{\Gamma}^2=\check{\Gamma}^3=\check{\Gamma}^4=\Q=\Q^{\pm},\label{n6_2}
\end{eqnarray}
and 
\begin{eqnarray*}
\Gamma^{\pm}\subset\Q^{\pm}\subset\P^{\pm}\subset\C^{\times}.
\end{eqnarray*}
\end{lemma}
\begin{proof}
We have the first two different groups in (\ref{n6_1}) by Lemma \ref{lemmag0'gn'} and the statement (\ref{g0gnoninner}).
We have $\Gamma^{1}=\Gamma^{5}=\check{\Gamma}^1=\check{\Gamma}^5$, $\Gamma^{2}=\Gamma^{4}=\check{\Gamma}^2=\check{\Gamma}^4$ by Lemma \ref{gkgnk}, by the remark after it, and by (\ref{gkgkeven}). We obtain  $\Gamma^{3}=\Gamma^{\overline{3}}=\Q=\check{\Gamma}^3=\check{\Gamma}^{\overline{3}}=\Q^{\pm}$ using Theorem 6.3 from \cite{OnInner}, Theorem \ref{maintheoremQ*}, and the statement (\ref{q=q}).

Let us prove $\check{\Gamma}^2=\Q^{\pm}$. We have $\check{\Gamma}^2\subseteq\Q^{\pm}$ by Lemma \ref{gkq*gkq*gkq'}. We need to prove $\Q^{\pm}\subseteq\check{\Gamma}^2$. Suppose $T\in\Q^{\pm}$. Since $\Q^{\pm}=\check{\Gamma}^{\overline{2}}$ by Theorem \ref{maintheoremQ*}, we have $T\C^{2}T^{-1}\subseteq\C^{\overline{2}}=\C^2\oplus\C^{6}$. Suppose that
\begin{eqnarray*}
T U_{2}T^{-1}= V_2 +\lambda e_{1\ldots6},\qquad U_2,V_2\in\C^{2},\qquad \lambda\in\F.
\end{eqnarray*}
Then we get
\begin{eqnarray*}
\lambda e &=& (e_{1\ldots6})^{-1}T U_2 T^{-1} - (e_{1\ldots6})^{-1}V_2=\langle (e_{1\ldots6})^{-1}T U_2 T^{-1} - (e_{1\ldots6})^{-1}V_2\rangle_0
\\
&=& \langle (e_{1\ldots6})^{-1}T U_2 T^{-1}\rangle_0=\pm\langle e_{1\ldots6} U_2\rangle_0=0
\end{eqnarray*}
using the property $\langle AB\rangle_0=\langle BA\rangle_0$ of the projection $\langle\;\;\rangle_0$ onto the subspace $\C^{0}$ and the fact that the element  $e_{1\ldots6}$ commutes with all even elements and anticommutes with all odd elements. Thus, $T\C^{2}T^{-1}\subseteq\C^{2}$ and the proof is completed. The four considered groups are different by Lemma \ref{glq*} and Lemma \ref{lemmaQABP}.
\end{proof}

Let us write down all of the groups considered in the paper \cite{OnInner} and in this paper in the cases of small dimensions $n\leq6$.

If $n=1$, then we have the following two different groups
\begin{eqnarray*}
&&\Gamma^{\overline{0}}=\Gamma^{0}=\Gamma^{\overline{1}}=\Gamma^{1}=\Gamma^{\overline{01}}=\Gamma^{01}=\check{\Gamma}^{\overline{01}}=\check{\Gamma}^{01}=\Gamma^{\overline{03}}=\Gamma^{\overline{12}} =\Gamma^{(0)}=\Gamma^{(1)}=\Gamma=\P=\A=\B=\B_{\pm}=\Q=\Q'=\C^{\times},
\\
&&\check{\Gamma}^{\overline{0}}=\check{\Gamma}^{0}=\check{\Gamma}^{\overline{1}}=\check{\Gamma}^1=\check{\Gamma}^{\overline{03}}=\check{\Gamma}^{\overline{12}}=\check{\Gamma}^{(0)}=\check{\Gamma}^{(1)}=\Gamma^{\pm}=\P^{\pm}=\A_{\pm}=\Q^{\pm}.
\end{eqnarray*}

If $n=2$, then we have the following two different groups
\begin{eqnarray*}
&&\Gamma^{\overline{0}}=\Gamma^{0}=\check{\Gamma}^{\overline{2}}=\check{\Gamma}^2=\Gamma^{\overline{3}}=\check{\Gamma}^{\overline{3}}=\check{\Gamma}^{\overline{01}}=\check{\Gamma}^{\overline{23}}=\Gamma^{\overline{03}}=\Gamma^{\overline{12}}=\B_{\pm}=\B=\C^{\times},
\\
&&\check{\Gamma}^{\overline{0}}=\check{\Gamma}^{0}=\check{\Gamma}^{\overline{1}}=\check{\Gamma}^1=\Gamma^{\overline{1}}=\Gamma^{1}=\Gamma^{\overline{2}}=\Gamma^{2}=\check{\Gamma}^{02}=\Gamma^{02}=\check{\Gamma}^{\overline{12}}=\check{\Gamma}^{\overline{03}}=\Gamma^{\overline{23}}=\Gamma^{\overline{01}}
\\
&&=\check{\Gamma}^{(0)}=\check{\Gamma}^{(1)}=\Gamma^{(0)}=\Gamma^{(1)}=\Gamma^{\pm}=\Gamma=\P^{\pm}=\P=\A_{\pm}=\A=\Q^{\pm}=\Q=\Q'.
\end{eqnarray*}

If $n=3$, then we have the following four different groups
\begin{eqnarray*}
&&\Gamma^{\overline{0}}=\Gamma^{0}=\Gamma^{\overline{3}}=\Gamma^{3}=\Gamma^{03}=\Gamma^{\overline{03}}=\Gamma^{\overline{12}}=\B=\C^{\times},
\\
&&\check{\Gamma}^{\overline{0}}=\check{\Gamma}^{0}=\check{\Gamma}^{\overline{1}}=\check{\Gamma}^1=\check{\Gamma}^{\overline{2}}=\check{\Gamma}^2=\check{\Gamma}^{\overline{3}}=\check{\Gamma}^3=\check{\Gamma}^{(0)}=\check{\Gamma}^{(1)}=\Gamma^{\pm}=\P^{\pm}=\Q^{\pm},
\\
&&\check{\Gamma}^{\overline{01}}=\check{\Gamma}^{\overline{23}}=\B_{\pm},
\\
&&\Gamma^{\overline{1}}=\Gamma^{1}=\Gamma^{\overline{2}}=\Gamma^{2}=\check{\Gamma}^{03}=\Gamma^{\overline{23}}=\Gamma^{\overline{01}}=\check{\Gamma}^{\overline{03}}=\check{\Gamma}^{\overline{12}}=\Gamma^{(0)}=\Gamma^{(1)}=\Gamma=\Q=\Q'=\P=\A=\A_{\pm}=\Z^{\times}\C^{0\times}.
\end{eqnarray*}

If $n=4$, then we have the following five different groups
\begin{eqnarray*}
&&\check{\Gamma}^4=\Gamma^{0}=\C^{\times},
\\
&&\check{\Gamma}^{\overline{0}}=\check{\Gamma}^{0}=\Gamma^{\overline{0}}=\Gamma^{4}=\check{\Gamma}^{\overline{2}}=\check{\Gamma}^2=\Gamma^{\overline{2}}=\Gamma^{2}=\check{\Gamma}^{04}=\Gamma^{04}=\check{\Gamma}^{(0)}=\check{\Gamma}^{(1)}=\Gamma^{(0)}=\Gamma^{(1)}=\P^{\pm}=\P=\Q',
\\
&&\check{\Gamma}^{\overline{1}}=\check{\Gamma}^1=\Gamma^{\overline{1}}=\Gamma^{1}=\check{\Gamma}^{\overline{3}}=\check{\Gamma}^3=\Gamma^{\overline{3}}=\Gamma^{3}=\Gamma^{\pm}=\Gamma=\Q^{\pm}=\Q,
\\
&&\check{\Gamma}^{\overline{03}}=\check{\Gamma}^{\overline{12}}=\Gamma^{\overline{01}}=\Gamma^{\overline{23}}=\A_{\pm}=\A,
\\
&&\check{\Gamma}^{\overline{01}}=\check{\Gamma}^{\overline{23}}=\Gamma^{\overline{03}}=\Gamma^{\overline{12}}=\B_{\pm}=\B.
\end{eqnarray*}

If $n=5$, then we have the following eight different groups
\begin{eqnarray*}
&&\Gamma^{0}=\Gamma^{5}=\Gamma^{05}=\C^{\times},
\\
&&\check{\Gamma}^{0}=\check{\Gamma}^5=\check{\Gamma}^{(0)}=\check{\Gamma}^{(1)}=\P^{\pm},
\\
&&\check{\Gamma}^{05}=\Gamma^{(0)}=\Gamma^{(1)}=\P=\Z^{\times}\C^{(0)\times},
\\
&&\check{\Gamma}^{\overline{0}}=\check{\Gamma}^{\overline{1}}=\check{\Gamma}^1=\Gamma^{\overline{2}}=\check{\Gamma}^2=\check{\Gamma}^{\overline{3}}=\check{\Gamma}^3=\check{\Gamma}^4=\Gamma^{\pm}=\Q^{\pm},
\\
&&\Gamma^{\overline{0}}=\Gamma^{4}=\Gamma^{\overline{1}}=\Gamma^{1}=\Gamma^{\overline{2}}=\Gamma^{2}=\Gamma^{\overline{3}}=\Gamma^{3}=\Gamma=\Q=\Q',
\\
&&\check{\Gamma}^{\overline{03}}=\check{\Gamma}^{\overline{12}}=\A_{\pm},
\\
&&\Gamma^{\overline{23}}=\Gamma^{\overline{01}}=\A,
\\
&&\check{\Gamma}^{\overline{01}}=\check{\Gamma}^{\overline{23}}=\Gamma^{\overline{03}}=\Gamma^{\overline{12}}=\B_{\pm}=\B.
\end{eqnarray*}

If $n=6$, then we have the following six different groups
\begin{eqnarray*}
&&\check{\Gamma}^6=\Gamma^{0}=\C^{\times},
\\
&&\check{\Gamma}^{0}=\Gamma^{6}=\check{\Gamma}^{06}=\Gamma^{06}=\check{\Gamma}^{(0)}=\check{\Gamma}^{(1)}=\Gamma^{(0)}=\Gamma^{(1)}=\P^{\pm}=\P,
\\
&&\check{\Gamma}^1=\check{\Gamma}^5=\Gamma^{1}=\Gamma^{5}=\Gamma^{\pm}=\Gamma,
\\
&&\check{\Gamma}^{\overline{0}}=\check{\Gamma}^4=\Gamma^{\overline{0}}=\Gamma^{4}=\check{\Gamma}^{\overline{1}}=\Gamma^{\overline{1}}=\check{\Gamma}^{\overline{2}}=\check{\Gamma}^2=\Gamma^{\overline{2}}=\Gamma^{2}=\check{\Gamma}^{\overline{3}}=\check{\Gamma}^3=\Gamma^{\overline{3}}=\Gamma^{3}=\Q^{\pm}=\Q=\Q',
\\
&&\check{\Gamma}^{\overline{03}}=\check{\Gamma}^{\overline{12}}=\Gamma^{\overline{01}}=\Gamma^{\overline{23}}=\A_{\pm}=\A,
\\
&&\check{\Gamma}^{\overline{01}}=\check{\Gamma}^{\overline{23}}=\Gamma^{\overline{03}}=\Gamma^{\overline{12}}=\B_{\pm}=\B.
\end{eqnarray*}

\section{The generalized spin groups}\label{sectGS}

Let us consider the well-known Lipschitz group
\begin{eqnarray}
\Gamma^\pm&=& \{ T\in\C^\times:\quad \widehat{T} \C^1 T^{-1}\subseteq \C^1\}\label{Lip11}\\
&=&\{T\in\C^{(0)\times}\cup\C^{(1)\times}:\quad T \C^1 T^{-1}\subseteq \C^1\}\label{Lip21}\\
&=&\{ v_1 \cdots v_m:\quad m\leq n,\quad v_j\in\C^{ 1\times}\}.\label{Lip31}
\end{eqnarray}
We have the following well-known even subgroup of $\Gamma^\pm$:
\begin{eqnarray}
\Gamma^+&=& \{ T\in\C^{(0)\times}:\quad T \C^1 T^{-1}\subseteq \C^1\}\label{Lip1+}\\
&=&\{ v_1 \cdots v_{2m}:\quad 2m\leq n,\quad v_j\in\C^{ 1\times}\}.\label{Lip3+}
\end{eqnarray}
The symbol $+$ in the notation of the group $\Gamma^+$ means that this group consists only of even elements. In the same way, we can consider the even subgroup of the group $\P^\pm=\C^{(0)\times}\cup\C^{(1)\times}$:
\begin{eqnarray}
\P^+&=&\C^{(0)\times}\subset\P^\pm.\label{P+}
\end{eqnarray}
Let us consider the following subgroups of the groups 
\begin{eqnarray}
\A_{\pm}=\{T\in\C^\times:\quad\widetilde{T}T\in \C^{ 0\times}\},\qquad
\B_{\pm}=\{T\in\C^\times:\quad\widehat{\widetilde{T}}T\in \C^{0\times }\}\label{AB1}
\end{eqnarray}
respectively:
\begin{eqnarray}
\A_+=\{T\in\C^\times:\quad\widetilde{T}T>0\}\subset \A_\pm,\qquad \B_+=\{T\in\C^\times:\quad\widehat{\widetilde{T}}T>0\}\subset\B_\pm.\label{AB+}
\end{eqnarray}
The symbol $+$ in the notation of these two groups means that these groups consist only of elements with positive values of the corresponding norm functions $\widetilde{T}T$ and $\widehat{\widetilde{T}}T$.

We also introduce the following subgroups of the group
\begin{eqnarray}
\Q^{\pm}=\{T\in\C^{(0)\times}\cup\C^{(1)\times}:\quad \widetilde{T}T\in\C^{ 0\times}\}\label{Q1}
\end{eqnarray}
(compare with the notation for spin groups in \cite{lg1}):
\begin{eqnarray}
\Q^{\pm}_{+\A}&=&\{T\in\C^{(0)\times}\cup\C^{(1)\times}:\quad \widetilde{T}T>0\}\subset \A_+,\label{Qpm+A}\\ \Q^{\pm}_{+\B}&=&\{T\in\C^{(0)\times}\cup\C^{(1)\times}:\quad \widehat{\widetilde{T}}T>0\}\label{Qpm+B}\subset\B_+,\\
\Q^{+}_\pm&=&\{T\in\C^{(0)\times}:\quad \widetilde{T}T\in\C^{ 0\times}\}=\{T\in\C^{(0)\times}:\quad \widehat{\widetilde{T}}T\in\C^{ 0\times}\}\subset\P^+,\label{Q+pm}\\\Q^{+}_+&=&\{T\in\C^{(0)\times}:\quad \widetilde{T}T>0\}=\{T\in\C^{(0)\times}:\quad \widehat{\widetilde{T}}T>0\}\subset\P^+.\label{Q++}
\end{eqnarray}
The standard spin groups are defined as normalized subgroups of the Lipschitz group $\Gamma^\pm$ (\ref{Lip1}) and its even subgroup $\Gamma^+$ (\ref{Lip1+}):
\begin{eqnarray}
\Pin&=&\{T\in\Gamma^\pm:\quad \widetilde{T} T=\pm e\}=\{T\in\Gamma^\pm:\quad \widehat{\widetilde{T}} T=\pm e\},\label{Pin}\\
\Pin_{+\A}&=&\{T\in\Gamma^\pm:\quad \widetilde{T}T=+e\},\\
\Pin_{+\B}&=&\{T\in\Gamma^\pm:\quad \widehat{\widetilde{T}}T=+e\},\\
\Spin&=&\{T\in\Gamma^+:\quad \widetilde{T}T=\pm e\}=\{T\in\Gamma^+:\quad \widehat{\widetilde{T}}T=\pm e\},\\
\Spin_+&=&\{T\in\Gamma^+:\quad \widetilde{T}T=+e\}=\{T\in\Gamma^+:\quad \widehat{\widetilde{T}}T=+e\}.\label{Spin+}
\end{eqnarray}
In the real case $\F=\R$, the spin groups (\ref{Pin}) - (\ref{Spin+}) are two-sheeted coverings of the corresponding orthogonal group $\OO(p,q)$, orthochorous (or parity preserving) group $\OO_{+\A}(p,q)$, orthochronous group $\OO_{+\B}(p,q)$, special (or proper) orthogonal group $\SO(p,q)$, and special orthochronous group $\SO_+(p,q)$ respectively under the twisted adjoint action $\check{\ad}$. Namely, for any matrix $P=(p_a^b)\in\OO(p,q)$, there exist exactly two elements $\pm T\in\Pin(p,q)$ in the corresponding spin group that are related in the following way (and similarly for the other orthogonal groups and spin groups respectively):
$$\widehat{T} e_a T^{-1}=p_a^b e_b.$$

We can consider the following normalized subgroups of the groups $\A_\pm$ and $\B_\pm$ (\ref{AB1}):
\begin{eqnarray}
\Pin^\A&=&\{T\in\C^\times:\quad \widetilde{T} T=\pm e\}\subset \A_\pm,\\
\Pin^\B&=&\{T\in\C^\times:\quad \widehat{\widetilde{T}}T=\pm e\}\subset \B_\pm,\\
\Pin^\A_{+}&=&\{T\in\C^\times:\quad \widetilde{T} T=+e\}\subset \A_+,\\
\Pin^\B_{+}&=&\{T\in\C^\times:\quad \widehat{\widetilde{T}}T=+e\}\subset \B_+.
\end{eqnarray}
The groups $\Pin^\A_{+}$ and $\Pin^\B_{+}$ are considered in \cite{Lie1, Lie2, Lie3}, where they are denoted by $\G^{23}$ and $\G^{12}$ in the real case $\F=\R$ and $\G^{23i23}$ and $\G^{12i12}$ in the complex case $\F=\CC$ respectively. In these three papers, isomorphisms between  $\Pin^\A_{+}$, $\Pin^\B_{+}$, and classical matrix Lie groups are proved in the case of arbitrary dimension and signature.

We can also consider the following normalized subgroups of the group $\Q^\pm$ (\ref{Q1}) and its subgroups (\ref{Qpm+A}) - (\ref{Q++}):
\begin{eqnarray}
\Pin^\Q&=&\{T\in\C^{(0)\times}\cup\C^{(1)\times}:\quad \widetilde{T} T=\pm e\}=\{T\in\C^{(0)\times}\cup\C^{(1)\times}:\quad \widehat{\widetilde{T}} T=\pm e\}\subset \Q^\pm,\label{PinQ}\\
\Pin^\Q_{+\A}&=&\{T\in\C^{(0)\times}\cup\C^{(1)\times}:\quad \widetilde{T}T=+e\}\subset \Q^\pm_{+\A},\\
\Pin^\Q_{+\B}&=&\{T\in\C^{(0)\times}\cup\C^{(1)\times}:\quad \widehat{\widetilde{T}}T=+e\}\subset \Q^\pm_{+\B},\\
\Spin^\Q&=&\{T\in\C^{(0)\times}:\quad \widetilde{T}T=\pm e\}=\{T\in\C^{(0)\times}:\quad \widehat{\widetilde{T}}T=\pm e\}\subset \Q^+_\pm,\\
\Spin^\Q_+&=&\{T\in\C^{(0)\times}:\quad \widetilde{T}T=+e\}=\{T\in\C^{(0)\times}:\quad \widehat{\widetilde{T}}T=+e\}\subset \Q^+_+.\label{SpinQ+}
\end{eqnarray}
We call the groups (\ref{PinQ}) - (\ref{SpinQ+}) \textit{generalized spin groups}. The group $\Spin^\Q_+$ is also considered in \cite{Lie1, Lie2, Lie3}, where it is denoted by $\G^2$ in the real case $\F=\R$ and $\G^{2i2}$ in the complex case $\F=\CC$. In these three papers, isomorphisms between $\Spin^\Q_+$ and classical matrix Lie groups are proved in the case of arbitrary dimension and signature.

Note that the generalized spin groups (\ref{PinQ}) - (\ref{SpinQ+}) coincide with the corresponding spin groups (\ref{Pin}) - (\ref{Spin+}) in the cases of small dimensions $n\leq 5$. In the cases $n\geq 6$, spin groups are subgroups of the corresponding generalized spin groups. Some of the considered groups are related to automorphism groups of the scalar products on the spinor spaces \cite{lounesto, Port, lg1}. The relation between the generalized spin groups and the orthogonal groups (or their generalizations) in the cases $n\geq 6$ requires further research.

\section{The corresponding Lie algebras}\label{liealg}

The group $\C^\times$ of all invertible elements of the geometric algebra $\C$ can be considered as a Lie group. The corresponding Lie algebra is the geometric algebra $\C$ with respect to the operation of commutator $[U, V]=UV-VU$ for $U, V\in\C$. We have $\dim{\C^\times}=2^n$.

It is well known
(see, for example, \cite{lg1, lounesto, lg2, HelmBook}) 
that the Lie algebra of the Lipschitz group $\check{\Gamma}^{1}$ is 
\begin{eqnarray}\label{gl}
\check{\gamma}:=\C^{0}\oplus\C^{2}
\end{eqnarray}
of dimension
\begin{eqnarray}\label{dimgl}
\dim{\check{\gamma}}=\dim{\check{\Gamma}^{1}}=1+\frac{n(n-1)}{2}.
\end{eqnarray}
Moreover, it is well known (see, for example, \cite{lg1, lounesto, lg2, gaforph}) that the Lie algebra of the spin groups $\Pin$, $\Pin_{+\A}$, $\Pin_{+\B}$, $\Spin$, $\Spin_{+}$, which are the subgroups of $\check{\Gamma}^{1}$, is 
\begin{eqnarray}\label{s}
\mathfrak{s}:= \C^{2}, \qquad \dim{\mathfrak{s}}=\frac{n(n-1)}{2}.
\end{eqnarray}
Let us write down in Table \ref{table1} the Lie groups considered in this paper and the corresponding Lie algebras with their dimensions. For the reader's convenience, we write down in this table the mentioned well-known facts (\ref{gl}) - (\ref{s}). The second column of the table is empty for the Lie groups that have the same Lie algebra for any natural $n\geq 1$.
\begin{table}
\caption{The Lie groups and the corresponding Lie algebras}
\centering
\begin{tabular}{ c c c c }
\toprule
Lie group & n & Lie algebra & Dimension \\ 
\midrule
$\C^{\times}$ & & $\C$ & $2^n$ \\ \hline
$\P^{\pm}=\check{\Gamma}^{(0)}=\check{\Gamma}^{(1)}$ & & $\C^{(0)}$ & $2^{n-1}$ \\ 
$\check{\Gamma}^{0n}$ & $0\mod{2}$ &  &
\\ \hline
$\check{\Gamma}^{0n}$ & $1\mod{2}$ & $\C^{(0)n}$ & $1 + 2^{n-1}$\\ \hline
$\A_{\pm}=\check{\Gamma}^{\overline{12}}=\check{\Gamma}^{\overline{03}}$, $\quad\A_{+}$ & & $\C^{0\overline{2}\overline{3}}$ & $1 + 2^{n-1}-2^{\frac{n}{2}-1}(\sin{(\frac{\pi n}{4})} + \cos{(\frac{\pi n}{4})})$ \\ \hline
$\B_{\pm}=\check{\Gamma}^{\overline{01}}=\check{\Gamma}^{\overline{23}}$, $\quad\B_{+}$ & & $\C^{0\overline{12}}$ & $1+2^{n-1}+2^{\frac{n}{2}-1}(\sin{(\frac{\pi n}{4})} - \cos{(\frac{\pi n}{4})})$ \\ \hline
  $\Q^{\pm}$, $\Q^{+}_{\pm}$, $\Q^{\pm}_{+\A}$, $\Q^{\pm}_{+\B}$, $\Q^{+}_{+}$ ${\;}^{\dagger}$ &
 & $\C^{0\overline{2}}$ & $1 + 2^{n-2}-2^{\frac{n}{2}-1}\cos{(\frac{\pi n}{4})}$ \\ 
$\Q'$ & $2\mod{4}$ &  &  \\ 
  \hline
$\Q'$ & $0,1,3\mod{4}$ & $\C^{0\overline{2}n}$ & $2 + 2^{n-2}-2^{\frac{n}{2}-1}\cos{(\frac{\pi n}{4})}$  \\ \hline
$\A'$ & $0,1\mod{4}$ & $\C^{0\overline{23}n}$ & $2 + 2^{n-1}-2^{\frac{n}{2}-1}(\sin{(\frac{\pi n}{4})} + \cos{(\frac{\pi n}{4})})$ \\ 
 & $2,3\mod{4}$ & $\C^{0\overline{23}}$ & $1 + 2^{n-1}-2^{\frac{n}{2}-1}(\sin{(\frac{\pi n}{4})} + \cos{(\frac{\pi n}{4})})$ \\ \hline
$\B'$ & $0,3\mod{4}$ & $\C^{0\overline{12}n}$ & $2+2^{n-1}+2^{\frac{n}{2}-1}(\sin{(\frac{\pi n}{4})} - \cos{(\frac{\pi n}{4})})$ \\ 
 & $1,2\mod{4}$ & $\C^{0\overline{12}}$ & $1+2^{n-1}+2^{\frac{n}{2}-1}(\sin{(\frac{\pi n}{4})} - \cos{(\frac{\pi n}{4})})$ \\ \hline
 $\check{\Gamma}^{1}$, $\Gamma^{+}$ & & $\C^{02}$ & $1 + \frac{n(n-1)}{2}$ \\ \hline
 $\Pin$, $\Pin_{+\A}$, $\Pin_{+\B}$, $\Spin$, $\Spin_{+}$  & & $\C^{2}$ & $\frac{n(n-1)}{2}$ \\
 \hline
 $\Pin^{\Q}$, $\Pin^{\Q}_{+\A}$, $\Pin^{\Q}_{+\B}$, $\Spin^{\Q}$, $\Spin^{\Q}_{+}$ & & $\C^{\overline{2}}$ &  $2^{n-1}-2^{\frac{n}{2}-1}(\sin{(\frac{\pi n}{4})} + \cos{(\frac{\pi n}{4})})$ \\ \hline 
 $\Pin^{\A}$, $\Pin^{\A}_{+}$ & & $\C^{\overline{23}}$ & $ 2^{n-1}-2^{\frac{n}{2}-1}(\sin{(\frac{\pi n}{4})} + \cos{(\frac{\pi n}{4})})$\\ 
 \hline
   $\Pin^{\B}$, $\Pin^{\B}_{+}$ & & $\C^{\overline{12}}$ & $2^{n-1}+2^{\frac{n}{2}-1}(\sin{(\frac{\pi n}{4})} - \cos{(\frac{\pi n}{4})})$\\ 
 \bottomrule\label{table1}
\end{tabular}
 \begin{tablenotes}
${\;}^{\dagger}$  Note that in the cases $n\geq3$, we have $\Q^{\pm}=\check{\Gamma}^{\overline{0}}=\check{\Gamma}^{\overline{1}}=\check{\Gamma}^{\overline{2}}=\check{\Gamma}^{\overline{3}}$ if $n=1,2,3\mod{4}$ and $\Q'=\check{\Gamma}^{\overline{0}}=\check{\Gamma}^{\overline{2}}$, $\Q^{\pm}=\check{\Gamma}^{\overline{1}}=\check{\Gamma}^{\overline{3}}$ if $n=0\mod{4}$ (see Theorem \ref{maintheoremQ*}). In the cases $n=1,2$, we have $\Q^{\pm}=\check{\Gamma}^{\overline{0}}=\check{\Gamma}^{\overline{1}}$ and $\check{\Gamma}^{\overline{2}}=\check{\Gamma}^{\overline{3}}=\C^{\times}$ (see Remark \ref{q*qsmall}).
\end{tablenotes}
\end{table}

\begin{theorem}
Let us denote the Lie algebras of the Lie groups $\P^{\pm}$, $\A_{\pm}$, $\B_{\pm}$, $\Q^{\pm}$, $\Q'$, $\A'$, and $\B'$ by 
\begin{eqnarray*}
&&\mathfrak{p}^{\pm} :=\C^{(0)},\qquad \mathfrak{a_{\pm}} :=\C^{0\overline{23}}, \qquad  \mathfrak{b_{\pm}} :=\C^{0\overline{12}}, \qquad \mathfrak{q^{\pm}}:=\C^{0\overline{2}},
\\
&&\mathfrak{q'}:=
\left\lbrace
\begin{array}{lll}
\C^{0\overline{2}n},\quad &n=0,1,3\mod{4},&
\\
\C^{0\overline{2}},\quad &n=2\mod{4},&
\end{array}
\right.
\quad 
\mathfrak{a'}:=
\left\lbrace
\begin{array}{lll}
\C^{0\overline{23}n},\quad &n=0,1\mod{4},&
\\
\C^{0\overline{23}},\quad &n=2,3\mod{4},&
\end{array}
\right.
\quad 
\mathfrak{b'}:=
\left\lbrace
\begin{array}{lll}
\C^{0\overline{12}n},\quad &n=0,3\mod{4},&
\\
\C^{0\overline{12}},\quad &n=1,2\mod{4}&
\end{array}
\right.
\end{eqnarray*}
respectively. We use the following notation
\begin{eqnarray*}
\mathfrak{s}^{\Q}:=\C^{\overline{2}},\qquad \mathfrak{s}^{\A}:=\C^{\overline{23}},\qquad \mathfrak{s}^{\B}:=\C^{\overline{12}}
\end{eqnarray*}
for the Lie algebras of the five Lie groups
$\Pin^{\Q},\Pin^{\Q}_{+\A},\Pin^{\Q}_{+\B},\Spin^{\Q}, \Spin^{\Q}_{+}$; the two Lie groups $\Pin^{\A},\Pin^{\A}_{+}$, and the two Lie groups $\Pin^{\B},\Pin^{\B}_{+}$
respectively.

The Lie algebras of the considered Lie groups are presented in Table \ref{table1} with corresponding dimensions.
We have
\begin{eqnarray}
&\check{\gamma}\subseteq\mathfrak{q^{\pm}}\subseteq\mathfrak{p^{\pm}},\qquad \mathfrak{q^{\pm}}\subseteq\mathfrak{a_{\pm}},\qquad \mathfrak{q^{\pm}}\subset\mathfrak{b_{\pm}},\qquad \mathfrak{q^{\pm}}=\mathfrak{a_{\pm}}\cap\mathfrak{p^{\pm}}=\mathfrak{b_{\pm}}\cap\mathfrak{p^{\pm}}=\mathfrak{a_{\pm}}\cap\mathfrak{b_{\pm}};\label{alg1}
\\
&\mathfrak{q^{\pm}}=\mathfrak{p^{\pm}}=\mathfrak{a_{\pm}}\neq\mathfrak{b_{\pm}},\qquad n=1,2;\qquad \mathfrak{q^{\pm}}=\mathfrak{p^{\pm}}\neq\mathfrak{a_{\pm}}\neq\mathfrak{b_{\pm}},\qquad n=3;\label{alg2}
\\
&\mathfrak{q^{\pm}}\neq\mathfrak{p^{\pm}},\qquad n\geq4;\label{alg3}
\\
&\check{\gamma}={\mathfrak{q^{\pm}}},\qquad n\leq 5;\qquad \check{\gamma}\neq{\mathfrak{q^{\pm}}},\qquad n=6;\label{alg4}
\\
&\mathfrak{q^{\pm}}\subseteq\mathfrak{q'}\subseteq\mathfrak{p^{\pm}},\qquad n=0\mod{2};\qquad \mathfrak{q^{\pm}}=\mathfrak{q'},\qquad n=2\mod{4};\qquad\mathfrak{q'}=\mathfrak{p^{\pm}},\qquad n=4;\label{alg5}
\\
&\mathfrak{s}\subseteq\mathfrak{s}^{\Q}\subseteq\mathfrak{s}^{\A}\subseteq\mathfrak{a_{\pm}}\subseteq\mathfrak{a'},\qquad \mathfrak{s}\subseteq\mathfrak{s}^{\Q}\subseteq\mathfrak{s}^{\B}\subseteq\mathfrak{b_{\pm}}\subseteq\mathfrak{b'};\label{s1}
\\
&\mathfrak{s}=\mathfrak{s}^{\Q},\qquad n\leq5.\label{s2}
\end{eqnarray}
\end{theorem}
Note that the statements (\ref{alg1}) - (\ref{alg5}) for the Lie algebras are similar to the statements (\ref{f1gq}), (\ref{f2gq}), (\ref{QABPcap}), (\ref{QABPcap2}), (\ref{QABPcap3}), (\ref{QABPcap4}), (\ref{falp}), (\ref{Q*Q'P*P}), (\ref{Q'=P*=P}) for the corresponding Lie groups. Every relation between the Lie algebras in (\ref{s1}) and (\ref{s2}) is similar to the relation (or the relations) between the corresponding Lie groups.
\begin{proof}
We use the well-known facts about the relation between an arbitrary group Lie and the corresponding Lie algebra in order to prove these statements. We calculate the dimensions of the considered Lie algebras using (see, for example, \cite{Lie3})
\begin{eqnarray*}
&\dim{\C^{(0)}}=2^{n-1},\qquad \dim{\C^{k}}=\binom{n}{k}, 
\\
&\dim{\C^{\overline{1}}}=2^{n-2} + 2^{\frac{n}{2}-1}\sin{(\frac{\pi n}{4})},\qquad \dim{\C^{\overline{2}}}=2^{n-2} - 2^{\frac{n}{2}-1}\cos{(\frac{\pi n}{4})},\qquad \dim{\C^{\overline{3}}}=2^{n-2} - 2^{\frac{n}{2}-1}\sin{(\frac{\pi n}{4})},
\end{eqnarray*}
where $\binom{n}{k}=\frac{n!}{k!(n-k)!}$ is the binomial coefficient.
One can easily verify the statements (\ref{alg1}) - (\ref{s2}) using the definitions of the corresponding Lie algebras.
\end{proof}

\section{Conclusions}\label{sectConcl}
In this paper, we study the Lie groups preserving the various fixed subspaces under the twisted adjoint representation in the case of the real and complex geometric algebras, and give their several equivalent definitions. 
These groups are interesting for consideration because the twisted adjoint representation is an important mathematical notion that is used to describe the relation between the orthogonal groups and the corresponding spin groups.
The Lie groups introduced in this paper can be interpreted as generalizations of the spin groups and the Lipschitz groups. The Lipschitz group $\check{\Gamma}^1$ and the standard spin groups $\Pin$, $\Pin_{+\A}$, $\Pin_{+\B}$, $\Spin$, $\Spin_{+}$ are subgroups of the groups considered in Sections \ref{sectionP*} - \ref{sectSmall}. Moreover, the Lipschitz group $\check{\Gamma}^1$ coincides with the group $\Q^{\pm}$ in the case $n\leq5$, with the group $\P^{\pm}$ in the case $n\leq3$, and with some other groups, depending on $n$ (see Section \ref{sectSmall}).

We consider the groups preserving the subspaces of fixed parity ($\check{\Gamma}^{(0)}=\check{\Gamma}^{(1)}=\P^{\pm}$, see Section \ref{sectionP*}), the subspaces of fixed quaternion types ($\check{\Gamma}^{\overline{m}}=\Q^{\pm}$ or $\Q'$, $m=0,1,2,3$, see Section \ref{grQ*}), the direct sum of the subspaces of fixed quaternion types  ($\check{\Gamma}^{\overline{03}}=\check{\Gamma}^{\overline{12}}=\A_{\pm}$ and $\check{\Gamma}^{\overline{01}}=\check{\Gamma}^{\overline{23}}=\B_{\pm}$, see Sections \ref{grA*} and \ref{grB*}), and the subspaces of fixed grades ($\check{\Gamma}^{k}$, $k=0,1\ldots n$, see Section \ref{sectiongk'}) under the twisted adjoint representation. Also, we study the groups that leave invariant the direct sum of the grade-$0$ and grade-$n$ subspaces under the adjoint representation and the twisted adjoint representation  ($\Gamma^{0n}=\P^{\pm}$ or $\C^{\times}$ and $\check{\Gamma}^{0n}=\P$ respectively, see Section \ref{sectionP*}). We study the Lie algebras of all these Lie groups and calculate their dimensions (the Lie groups with the corresponding Lie algebras and dimensions are presented in Section \ref{liealg} and Table \ref{table1}).
Using the definitions of the introduced groups $\P^{\pm}$, $\Q^{\pm}$, $\A_{\pm}$, and $\B_{\pm}$, we conclude that many of the considered groups, which leave invariant the fixed subspaces under the twisted adjoint representation, coincide in the cases of small dimension (see Section \ref{sectSmall}). We define all the different groups preserving the subspaces of fixed grades under the twisted adjoint representation in the case $n\leq6$ (and under the standard adjoint representation as well in the case $n=6$ for the first time). 
Moreover, we conclude that 
the groups preserving the fixed subspaces under the twisted adjoint representation are related with the groups preserving the fixed subspaces under the standard adjoint representation, which were considered in 
\cite{OnInner}. In the cases of small dimension ($n\leq6$), we study these relations in detail and write down all the different groups that we have. 
The groups $\P$, $\Q$, and $\Gamma^k$, $k=1,\ldots,n-1$, (in the case of arbitrary $n$) and the groups $\Gamma^{\overline{m}}$ (in the case $n\geq4$) coincide with the groups $\P^{\pm}$, $\Q^{\pm}$, $\check{\Gamma}^k$, and $\check{\Gamma}^{\overline{m}}$ respectively up to the multiplication by the invertible elements of the center $\Z^{\times}$ (see (\ref{rel_P}), (\ref{rel_Q}), (\ref{rel_gk}), and (\ref{rel_gk_ov})). The particular case $\Gamma^1=\Z^{\times}\check{\Gamma}^1$ of the statement (\ref{rel_gk}) is well known, and it describes the relation between the Clifford group $\Gamma^1$ and the Lipschitz group $\check{\Gamma}^1$. 

This paper presents the groups $\Pin^{\Q}$, $\Pin^{\Q}_{+\A}$, $\Pin^{\Q}_{+\B}$, $\Spin^{\Q}$, and $\Spin^{\Q}_{+}$, which we call the generalized spin groups (see Section \ref{sectGS}). Namely, the generalized spin groups are defined as normalized subgroups of $\Q^{\pm}$ and its subgroups $\Q^{\pm}_{+\A}$, $\Q^{\pm}_{+\B}$, $\Q^{+}_{\pm}$, and $\Q^{+}_{+}$. The generalized spin groups coincide with the corresponding standard spin groups in the cases of small dimensions $n \leq 5$. These groups can also be considered as subgroups of the groups $\Pin^{\A}$, $\Pin^{\B}$, $\Pin^{\A}_{+}$, and $\Pin^{\B}_{+}$, which are normalized subgroups of the groups $\A_{\pm}$ and $\B_{\pm}$. We study the Lie algebras of all the introduced Lie groups.  The relation between the generalized spin groups and the orthogonal groups (or their generalizations) in the cases $n\geq 6$ requires further research. Study of the corresponding Lie groups in degenerate geometric algebras $\C_{p,q,r}$, which are widely used in different applications, is another interesting problem for further research.
Being the generalizations of the standard Lipschitz and spin groups, the groups introduced in this paper may be useful for various applications in physics\cite{hestenes, gaforph, mph, Marchuk, toolforph}, engineering\cite{Hitzer, DL, CE, B2}, computer science \cite{Dorst, BC, BC2, DL, CE} (in particular, for neural networks \cite{neural_n} and image processing\cite{BB}), and other sciences.

\begin{ack}[Acknowledgements]
The authors are grateful to the organizers of the 8th Conference on Applied Geometric Algebras in Computer Science and Engineering (Brno, September 2021) and the participants of this conference for fruitful discussions.
\end{ack}

\end{document}